\documentclass[12pt,reqno]{amsart}
\allowdisplaybreaks[4]
\usepackage{graphicx} 
\usepackage{amssymb}
\usepackage{amsmath}
\usepackage{cite}
\usepackage{subfigure}
\usepackage{graphicx}
\newtheorem{proposition}{Proposition}
\newtheorem{lemma}{Lemma}

\begin{document}
\title[Rogue wave]{The hierarchy of higher order solutions of the derivative  nonlinear Schr\"odinger
equation}
\author{Yongshuai Zhang, Lijuan Guo, Shuwei Xu, Zhiwei Wu and Jingsong HE}
\dedicatory { Department of Mathematics, Ningbo University,
Ningbo , Zhejiang 315211, P.\ R.\ China }

\begin{abstract}
 In this paper, we provide a simple method to generate higher order position solutions and rogue wave solutions for the derivative nonlinear Schr\"odinger equation. The formulae of these higher order solutions
 are given in terms of determinants.
 The dynamics and structures of solutions generated by this method are studied.
\end{abstract}

\maketitle \vspace{2 ex}

\noindent {{\bf Keywords}: Higher-order positon. Higher-order rogue wave. Darboux transformation. Derivative  nonlinear Schr\"odinger equation.}


\section{{\bf Introduction}}

The derivative nonlinear Schr\"odinger equation (DNLS)
\begin{equation}\label{dnls}
{\rm i}q_{t}-q_{xx}+{\rm i}(|q|^2q)_{x}=0.
\end{equation}
plays an important role in plasma physics and nonlinear optics.
It not only dominates the evolution of small-amplitude Alf\'en waves in a
low-$\beta$ plasma\cite{JPSJ41265,JPP16321,PS40227,NWCSP}, but also is used to describe
the behavior of large-amplitude magnetohydrodynamic (MHD) waves in a high-$\beta$ plasma \cite{JPP67271,PLA3726107}. On the other hand, the DNLS equation governs the transmission of sub-picosecond in single mode optical fibers
\cite{PRA231266,PRA271393,NFO}.


For the DNLS equation with vanishing boundary condition (VBC), Kaup and Newell (KN)\cite{JMP19798} first found the
one-soliton solution by inverse scattering transformation (IST).  On the basis of bilinear transformation,
the first $N$-solition formula was obtained by Nakamuro and Chen \cite{JPSJ49813}. Determinant expression of the $N$-soliton solution can be established via applying the Darboux transformation \cite{JPA23439}. In the case of the
non-vanishing boundary condition(NVBC), Kawata and Inoue developed an IST for the DNLS equation and obtained a
breather-type soliton (paired soliton) \cite{JPSJ441968}.
 Wadati {\sl et al} derived the stationary
solution for the DNLS equation under the plane wave boundary \cite{JPSJ48279}.
Chen and Lam \cite{PRE69066604} revised the IST for the DNLS by introducing an affine parameter, and derived a breather solution, which can be degenerated to both dark soliton and bright soliton.


Recently, rogue wave, an emerging phenomenon, is passionately discussed. The concept of rogue wave
was first proposed in the studies of deep ocean waves \cite{NOWATIST,EOW}, and
gradually extended to other fields such as optics fibre\cite{Nature450,PRL101233902,OE163644},
 Bose-Einstein condensates \cite{PRA80033610},
 capillary phenomena \cite{PRL104104503}, and so on.
 Rogue wave,  ``which appears from nowhere and disappears without a trace (WANDT)'' \cite{PLA373675},
 possesses the following two remarkable characteristics: i) locates in both space and time, ii) exhibits a dominant peak.


The first order rogue wave  was found in 1983 by Peregrine \cite{JAMSS2516}, which is a solution of the NLS
equation. It is usually called the
Peregrine soliton, and has been observed
experimentally in fiber \cite{NP6790}, water tank \cite{PRL106204502} and  multi-component plasma \cite{PRL107255005}.
The first order rogue wave solution of  the DNLS equation was first found by Xu and coworkers \cite{JPA44305203} by
the Darboux transformation and certain limit technique.
Recently, Guo et al \cite{Guo} obtained two kinds of generalized Darboux
transformations, and got the formulae of higher order solutions
for both the VBC and NVBC. Moveover, Guo showed two patterns (fundamental
and triangular) of the second  order rogue wave from a special seed solution
$q={\rm exp}(-{\rm i}x)$,  which are similar to the case of the NLS equation.
 Therefore, it is nature to ask whether the rogue wave solution of the DNLS equation possesses new structures that have not been found in other soliton equations such as the NLS equation.


The Darboux transformation, generated from the work of Darboux in 1882 for the Sturm-Liouville equation, has been an important method in generating solutions of integrable systems. To get the rogue wave solution, we need to iterate Darboux transformation at the same eigenvalue, but it does not work in this case. So we must modify the  Darboux transformation to get the solutions at the same eigenvalue. In this paper, we adopt the Taylor expansion to
deal with this defect, and obtain positon solutions, rational traveling solutions and rogue wave solutions.


The organization of this paper is as follows. In section 2, we provide a new method to generate solutions at the same eigenvalue base on the method of Darboux transformation and Taylor expansion, and display the formula of $N$-th order solution in terms of determinant. As applications, several explicitly analytic expressions are given, which include positon solutions, rational traveling solutions and rogue wave solutions. In section 3, we obtain the multi-rogue wave solutions by altering the mixed coefficients of eigenfunctions, which contain several free parameters. With the help of those
parameters, we consider the dynamics of multi-rogue wave. Moreover, three kinds of new structures: \emph{modified-triangular} structure,  \emph{ring-triangle} structure and \emph{multi-ring} structure are given. The conclusion is given in the last section.

\section{{\bf The solutions of the DNLS equation}}

The Kaup and Newell (KN) system\cite{JMP19798}:
\begin{equation}
  \label{KN}
  \left\{
  \begin{aligned}
  r_t-{\rm i}r_{xx}-(r^2q)_x&=0,\\
  q_t+{\rm i}q_{xx}-(rq^2)_x&=0.
  \end{aligned}
  \right.
\end{equation}
can be represented as the integrability condition of the following Kaup and Newell spectral system (Lax pair)\cite{JMP19798,JPSJ68355}:


\begin{equation}\label{lax}
\left\{
\begin{aligned}
\Psi_{x}&=M\Psi=(J\lambda^2+Q\lambda)\Psi,\\
\Psi_{t}&=N\Psi=(2J\lambda^4+V_{3}\lambda^3+V_{2}\lambda^2+V_{1}\lambda)\Psi,
\end{aligned}
\right.
\end{equation}

with


\begin{equation}
{J=\left( \begin{array}{cc}
{\rm i} &0 \\
0 &-{\rm i}
\end{array}
\right),
\quad Q=\left( \begin{array}{cc}
0 &q\\
r &0
\end{array}
\right),} \nonumber
\end{equation}
\begin{equation}
V_{3}=2Q, \quad V_{2}=Jqr,
\quad V_{1}=\left(\begin{array}{cc}
0 &-{\rm i}q_{x}+q^2r \\
{\rm i}r_{x}+r^2q &0
\end{array}
\right), \nonumber
\end{equation}
Where $\lambda\in\mathbb{C}$,  $\Psi\in\mathbb{C}^2$,
$\Psi$ is called the eigenfunction of the spectral problem \eqref{lax} corresponding  to eigenvalue $\lambda$ .

When
\begin{equation}
  \label{reduction}
  r=-q^*,
\end{equation}
the KN system can be reduced to the DNLS equation, asterisk denotes complex conjugation.

\vspace{1 ex}
\noindent{\bf 2.1 Determinant expression}\\

\vspace{1 ex}
The $N$-th Darboux transformation of the KN system in terms of determinant was obtained In \cite{JPA44305203}.  And the formulae for $N$-th order solutions were given as following:


\begin{lemma}\cite{JPA44305203}
Let $\Psi_{i}=\left(\begin{array}{c}
f_i\\
g_i
\end{array}
\right)$ $\left(i=1,2,\cdots,n\right)$ be distinct solutions related to $\lambda_i$ of the spectral problem (\ref{lax}),
then ($q^{[n]}$,$r^{[n]}$)  given by the  following formulae are new solutions of the KN system \eqref{KN}.
\begin{equation}\label{q[n]}
q^{[n]}=\frac{\Omega_{11}^2}{\Omega_{21}^2}q+2{\rm i}\frac{\Omega_{11}\Omega_{12}}{\Omega_{21}^2},\quad
r^{[n]}=\frac{\Omega_{21}^2}{\Omega_{11}^2}r-2{\rm i}\frac{\Omega_{21}\Omega_{22}}{\Omega_{11}^2}.
\end{equation}
with  $n=2k$,
\begin{equation}
\quad\Omega_{11}=\begin{vmatrix}
\lambda_{1}^{n-1}g_1&\lambda_{1}^{n-2}f_{1}&\lambda_{1}^{n-3}g_{1}&\cdots&\lambda_{1}g_{1}&f_{1}\\
\lambda_{2}^{n-1}g_{2}&\lambda_{2}^{n-2}f_{2}&\lambda_{2}^{n-3}g_{2}&\cdots&\lambda_{2}g_{2}&f_{2}\\
\vdots&\vdots&\vdots&\vdots&\vdots&\vdots\\
\lambda_{n}^{n-1}g_{n}&\lambda_{n}^{n-2}f_{n}&\lambda_{n}^{n-3}g_{n}&\cdots&\lambda_{n}g_{n}&f_{n} \nonumber
\end{vmatrix},
\end{equation}
\begin{equation}
\quad\Omega_{12}=\begin{vmatrix}
\lambda_{1}^{n}f_1&\lambda_{1}^{n-2}f_{1}&\lambda_{1}^{n-3}g_{1}&\cdots&\lambda_{1}g_{1}&f_{1}\\
\lambda_{2}^{n}f_{2}&\lambda_{2}^{n-2}f_{2}&\lambda_{2}^{n-3}g_{2}&\cdots&\lambda_{2}g_{2}&f_{2}\\
\vdots&\vdots&\vdots&\vdots&\vdots&\vdots\\
\lambda_{n}^{n}f_{n}&\lambda_{n}^{n-2}f_{n}&\lambda_{n}^{n-3}g_{n}&\cdots&\lambda_{n}g_{n}&f_{n} \nonumber
\end{vmatrix},
\end{equation}
\begin{equation}
\quad\Omega_{21}=\begin{vmatrix}
\lambda_{1}^{n-1}f_1&\lambda_{1}^{n-2}g_{1}&\lambda_{1}^{n-3}f_{1}&\cdots&\lambda_{1}f_{1}&g_{1}\\
\lambda_{2}^{n-1}f_{2}&\lambda_{2}^{n-2}g_{2}&\lambda_{2}^{n-3}f_{2}&\cdots&\lambda_{2}f_{2}&g_{2}\\
\vdots&\vdots&\vdots&\vdots&\vdots&\vdots\\
\lambda_{n}^{n-1}f_{n}&\lambda_{n}^{n-2}g_{n}&\lambda_{n}^{n-3}f_{n}&\cdots&\lambda_{n}f_{n}&g_{n} \nonumber
\end{vmatrix},
\end{equation}
\begin{equation}{\label{even1}}
\Omega_{22}=\begin{vmatrix}\nonumber
\lambda_{1}^{n}g_1&\lambda_{1}^{n-2}g_{1}&\lambda_{1}^{n-3}f_{1}&\cdots&\lambda_{1}f_{1}&g_{1}\\
\lambda_{2}^{n}g_{2}&\lambda_{2}^{n-2}g_{2}&\lambda_{2}^{n-3}f_{2}&\cdots&\lambda_{2}f_{2}&g_{2}\\
\vdots&\vdots&\vdots&\vdots&\vdots&\vdots\\
\lambda_{n}^{n}g_{n}&\lambda_{n}^{n-2}g_{n}&\lambda_{n}^{n-3}f_{n}&\cdots&\lambda_{n}f_{n}&g_{n}
\end{vmatrix}.
\end{equation}
Here $\lambda_{2l}=-\lambda_{2l-1}^*$ and
$\Psi_{2l}=\left(\begin{array}{c}
      f_{2l}\\
      g_{2l}
      \end{array} \right)=\left(\begin{array}{c}
     g_{2l-1}^*\\
     f_{2l-1}^*
     \end{array}\right)$.
\end{lemma}

It is trivial to check $r^{[2k]}=-{q^{[2k]}}^*$.\\

\noindent{\bf 2.2 Solutions from vacuum}\\

Let us consider the trivial case. When $q=r=0$,  the following $\Psi$ is an eigenfunction for $\lambda$,
\begin{equation}\label{fun1}
\Psi=\left(\begin{array}{c}
f\\
g
\end{array}
\right),\quad
f={\rm exp}({\rm i}(\lambda^{2}x+2\lambda^{4}t)),\quad
g={\rm exp}(-{\rm i}(\lambda^{2}x+2\lambda^{4}t)).
\end{equation}\\


By applying the above formulae \eqref{q[n]}, we can get $N$-soliton solution of the DNLS equation from vacuum. To get new kinds of solutions, we set the eigenvalues share the same value, i.e., iterating the Darboux transformation at the same eigenvalue. However, the formulae \eqref{q[n]} will be ineffective in this case. Next, we will use the Taylor expansion to generate the Darboux
transformation and get the formula of $q^{[n]}$ at the same eigenvalue
as we have done for the case of the NLS equation\cite{arxiv12093742}.


At first, we define new functions $\Psi[i,j,k]$ for a general solution $\Psi=\Psi(\lambda)$ corresponding to $\lambda$ as following:
\begin{equation}
  \lambda^{j}\Psi=\Psi[i,j,0]+\Psi[i,j,1]\epsilon+\Psi[i,j,2]\epsilon^2+\cdots+\Psi[i,j,k]\epsilon^k+\cdots,
\end{equation}
with
$$\Psi[i,j,k]=\frac{1}{k!}\frac{\partial^k(\lambda_i^j\Psi(\lambda_i))}{\partial\lambda_i^k}.$$
In particular
$$\Psi[1,1,0]=\lambda_1\Psi(\lambda_1),\qquad
\Psi[i,j,0]=\lambda_i^j\Psi(\lambda_i).
$$
\begin{proposition}
{ Let $\lambda_1=\alpha_1+{\rm i}\beta_1$, $\lambda_2=-\lambda_1^*$, then the following formula is the $n$-th($n=2k$) solution of the DNLS equation generated at the same eigenvalue.
\begin{equation}\label{nposition}
 q^{[n]}=2{\rm i}\frac{\delta_{11}\delta_{12}}{\delta_{21}^2}
 \end{equation}
where
\begin{equation}
\delta_{11}=\begin{vmatrix}
g[1,n-1,0]&f[1,n-2,0]&g[1,n-3,0]&\cdots&g[1,1,0]&f[1,0,0]\\
g[2,n-1,0]&f[2,n-2,0]&g[2,n-3,0]&\cdots&g[2,1,0]&f[2,0,0]\\
g[1,n-1,1]&f[1,n-2,1]&g[1,n-3,1]&\cdots&g[1,1,1]&f[1,0,1]\\
g[2,n-1,1]&f[2,n-2,1]&g[2,n-3,1]&\cdots&g[2,1,1]&f[2,0,1]\\
\vdots&\vdots&\vdots&\vdots&\vdots&\vdots\\
g[1,n-1,k-1]&f[1,n-2,k-1]&g[1,n-3,k-1]&\cdots&g[1,1,k-1]&f[1,0,k-1]\\
g[2,n-1,k-1]&f[2,n-2,k-1]&g[2,n-3,k-1]&\cdots&g[2,1,k-1]&f[2,0,k-1]\nonumber
\end{vmatrix},
\end{equation}
\begin{equation}
\delta_{12}=\begin{vmatrix}
f[1,n,0]&f[1,n-2,0]&g[1,n-3,0]&\cdots&g[1,1,0]&f[1,0,0]\\
f[2,n,0]&f[2,n-2,0]&g[2,n-3,0]&\cdots&g[2,1,0]&f[2,0,0]\\
f[1,n,1]&f[1,n-2,1]&g[1,n-3,1]&\cdots&g[1,1,1]&f[1,0,1]\\
f[2,n,1]&f[2,n-2,1]&g[2,n-3,1]&\cdots&g[2,1,1]&f[2,0,1]\\
\vdots&\vdots&\vdots&\vdots&\vdots&\vdots\\
f[1,n,k-1]&f[1,n-2,k-1]&g[1,n-3,k-1]&\cdots&g[1,1,k-1]&f[1,0,k-1]\\
f[2,n,k-1]&f[2,n-2,k-1]&g[2,n-3,k-1]&\cdots&g[2,1,k-1]&f[2,0,k-1]\\\nonumber
\end{vmatrix},
\end{equation}
\begin{equation}
\delta_{21}=\begin{vmatrix}
f[1,n-1,0]&g[1,n-2,0]&f[1,n-3,0]&\cdots&f[1,1,0]&g[1,0,0]\\
f[2,n-1,0]&g[2,n-2,0]&f[2,n-3,0]&\cdots&f[2,1,0]&g[2,0,0]\\
f[1,n-1,1]&g[1,n-2,1]&f[1,n-3,1]&\cdots&f[1,1,1]&g[1,0,1]\\
f[2,n-1,1]&g[2,n-2,1]&f[2,n-3,1]&\cdots&f[2,1,1]&g[2,0,1]\\
\vdots&\vdots&\vdots&\vdots&\vdots&\vdots\\
f[1,n-1,k-1]&g[1,n-2,k-1]&f[1,n-3,k-1]&\cdots&f[1,1,k-1]&g[1,0,k-1]\\
f[2,n-1,k-1]&g[2,n-2,k-1]&f[2,n-3,k-1]&\cdots&f[2,1,k-1]&g[2,0,k-1]
\end{vmatrix}.\nonumber
\end{equation}}
\end{proposition}


\begin{proof}
For the entries in the first column of $\Omega_{11}$ \eqref{q[n]},
\begin{equation*}
\begin{aligned}
\lambda_1^{n-1}g_1&=g[1,n-1,0],\\
\lambda_2^{n-1}g_2&=g[2,n-1,0],\\
\lambda_3^{n-1}g_3&=g[1,n-1,0]+g[1,n-1,1]\epsilon,\\
\lambda_4^{n-1}g_4&=g[2,n-1,0]+g[2,n-1,1]\epsilon,\\
&\vdots\\
\lambda_{n-1}^{n-1}g_{n-1}&=g[1,n-1,0]+g[1,n-1,1]\epsilon+\cdots+g[1,n-1,k-1]\epsilon^{k-1},\\
\lambda_n^{n-1}g_n&=g[2,n-1,0]+g[2,n-1,1]\epsilon+\cdots+g[2,n-1,k-1]\epsilon^{k-1}.\\
\end{aligned}
\end{equation*}
Taking the similar procedure to the other entries in $\Omega_{11}$, $\Omega_{12}$, and $\Omega_{21}$. Finally, the $q^{[n]}$ can be obtained through simple calculation.
\end{proof}
For example, when $n=4$, \begin{equation}\label{q[41]}
  q^{[4]}=\frac{\delta_{11}^2}{\delta_{21}^2}q+2{\rm i}\frac{\delta_{11}\delta_{12}}{\delta_{21}^2}=2{\rm i}\frac{\delta_{11}\delta_{12}}{\delta_{21}^2},
\end{equation}
where
$$
\delta_{11}=\begin{vmatrix}
g[1,3,0]&f[1,2,0]&g[1,1,0]&f[1,0,0]\\
g[2,3,0]&f[2,2,0]&g[2,1,0]&f[2,0,0]\\
g[1,3,1]&f[1,2,1]&g[1,1,1]&f[1,0,1]\\
g[2,3,1]&f[2,2,1]&g[2,1,1]&f[2,0,1]
\end{vmatrix},
$$
$$
\delta_{12}=\begin{vmatrix}
f[1,4,0]&f[1,2,0]&g[1,1,0]&f[1,0,0]\\
f[2,4,0]&f[2,2,0]&g[2,1,0]&f[2,0,0]\\
f[1,4,1]&f[1,2,1]&g[1,1,1]&f[1,0,1]\\
f[2,4,1]&f[2,2,1]&g[2,1,1]&f[2,0,1]
\end{vmatrix},
$$
$$
\delta_{21}=\begin{vmatrix}
f[1,3,0]&g[1,2,0]&f[1,1,0]&g[1,0,0]\\
f[2,3,0]&g[2,2,0]&f[2,1,0]&g[2,0,0]\\
f[1,3,1]&g[1,2,1]&f[1,1,1]&g[1,0,1]\\
f[2,3,1]&g[2,2,1]&f[2,1,1]&g[2,0,1]
\end{vmatrix}.
$$
Substituting the eigenfunction \eqref{fun1} into the formula \eqref{q[41]}, we obtain the positon solution
\begin{equation}\label{posoliton1}
  q_{positon}^{[4]}=\frac{L_1^*L_2}{L_1^2},
\end{equation}
where
\begin{equation}\nonumber
\begin{split}
L_1=&G_1-{\rm i}G_2,\\
L_2=&-16{\rm i}\alpha_1\beta_1(\cos(F_2)+{\rm i}\sin(F_2))\left(\left(\beta_1^3+4{\rm i}\alpha_1^4\beta_1x+4{\rm i}\alpha_1^2\beta_1^3x
     -32{\rm i}\alpha_1^4\beta_1^3t+16{\rm i}\alpha_1^6\beta_1t-48{\rm i}\alpha_1^2 \right.\right.\\
    &\left.\times\beta_1^5t\right)\sinh(F_1)-\left.({\rm i}\alpha_1^3+4\alpha_1^3\beta_1^2x+32\alpha_1^3\beta_1^4t+48\alpha_1^5\beta_1^2t
    +4\beta_1^4\alpha_1x-16\beta_1^6\alpha_1t)\cosh(F_1)\right),\\
G_1=&{\alpha_{{1}}}^{4}+{\beta_{{1}}}^{4}+256\,{\alpha_{{1}}}^{8}{\beta_{{
     1}}}^{2}xt-256\,{\alpha_{{1}}}^{4}{\beta_{{1}}}^{6}xt+256\,{\alpha_{{1}
     }}^{6}{\beta_{{1}}}^{4}xt-256\,{\alpha_{{1}}}^{2}{\beta_{{1}}}^{8}xt+
     512\,{\alpha_{{1}}}^{2}{\beta_{{1}}}^{10}{t}^{2}\\
    &+32\,{\alpha_{{1}}}^{2
     }{\beta_{{1}}}^{6}{x}^{2}+32\,{\alpha_{{1}}}^{6}{\beta_{{1}}}^{2}{x}^{
     2}+512\,{\alpha_{{1}}}^{10}{\beta_{{1}}}^{2}{t}^{2}+2048\,{\alpha_{{1}
     }}^{8}{\beta_{{1}}}^{4}{t}^{2}+3072\,{\alpha_{{1}}}^{6}{\beta_{{1}}}^{
     6}{t}^{2}\\
    &+64\,{\alpha_{{1}}}^{4}{\beta_{{1}}}^{4}{x}^{2}+2048\,{\alpha
     _{{1}}}^{4}{\beta_{{1}}}^{8}{t}^{2}+({\alpha_{{1}}}^{4}-{\beta_{{1}}}^{4})\cosh \left( 2F_1 \right), \\
G_2=&-16\,{\alpha_{{1}}}^{2}{\beta_{{1}}}^{4}x-384\,{\alpha_{{1}}}^{4}{
     \beta_{{1}}}^{4}t+64\,{\alpha_{{1}}}^{2}{\beta_{{1}}}^{6}t+16\,{\alpha
     _{{1}}}^{4}{\beta_{{1}}}^{2}x+(2\,{\alpha_{{1}}}^{3}\beta_{{1}}+2\,\alpha_{{1}}{\beta_{{1}}}^{3})\sinh
     \left( 2F_1  \right) \\
     &+64\,{\alpha_{{1}}}^{6}{\beta_{{1}
     }}^{2}t,\\
F_1=&4\,\alpha_{{1}}\beta_{{1}} \left( 4\,t{\alpha_{{1}}}^{2}-4\,t{\beta_{{1}}}^{2}+x \right),\\
F_2=&2\,{\alpha_{{1}}}^{2}x+4\,{\alpha_{{1}}}^{4}t-24\,t{\alpha_{{1}}}^{2}{\beta_{{1}}}^{2}-2\,
     {\beta_{{1}}}^{2}x+4\,{\beta_{{1}}}^{4}t.
  \end{split}
\end{equation}

when $x\rightarrow\pm\infty$, $|q^{[4]}|=0$, when $x=0$, $t=0$, $|q^{[4]}|^2=64\beta_1^2$. A simple analysis shows that it possesses phase shift compared with $2$-rd soliton when $t\rightarrow\pm\infty$. After taking values as $\alpha_1=0.5$, $\beta_1=0.5$, the evolution of positon solution of the DNLS equation is shown in Fig. \ref{positon}.



Next, if we set $\alpha_1\rightarrow0$ in above procedure, we will get the  the second order rational traveling solution. With these parameters, we find that the general
solution can be given in the same form as \eqref{posoliton1}, but with the values for $L_1$ and $L_2$ written by

\begin{equation}\label{posoliton12}
  \begin{split}
    L_1=&G_1-{\rm i}G_2,\\
    L_2=&-8{\rm i}\beta_1{\rm exp}(2{\rm i}\beta_1^2(2\beta_1^2t-x))(-3{\rm i}-12\beta_1^2x-48\beta_1^4t+48{\rm i}
         \beta_1^4x^2+2304{\rm i}\beta_1^8t^2-768{\rm i}\beta_1^6xt\\
        &-64\beta_1^6x^3+4096\beta_1^{12}t^3-3072\beta_1^{10}xt^2+768\beta_1^8tx^2),\\
    G_1=&3+4096\,{\beta_{{1}}}^{10}t{x}^{3}-24576\,{\beta_{{1}}}^{12}{t}^{2}{x}
         ^{2}+65536\,x{\beta_{{1}}}^{14}{t}^{3}-768\,{\beta_{{1}}}^{6}xt+96\,{
         \beta_{{1}}}^{4}{x}^{2}-65536\,{\beta_{{1}}}^{16}{t}^{4}\\
         &+4608\,{\beta_
         {{1}}}^{8}{t}^{2}-256\,{\beta_{{1}}}^{8}{x}^{4},\\
    G_2=&576\,{\beta_{{1}}}^{4}t-48\,{\beta_{{1}}}^{2}x+3072\,{\beta_{{1}}}^{8}
         t{x}^{2}-256\,{\beta_{{1}}}^{6}{x}^{3}+16384\,{\beta_{{1}}}^{12}{t}^{3
         }-12288\,{\beta_{{1}}}^{10}x{t}^{2}.
  \end{split}\nonumber
\end{equation}


The dynamics of rational travelling solution of the DNLS equation with $\beta_1=0.3$ are shown in Fig. \ref{ration}. Actually, it represents the interaction of two rational traveling solitons, and shares same properties with positon. \\




\noindent
{\bf 2.3. Solutions from periodic solution}\\


Here, we will apply the method discussed above to generate solutions from periodic seed solution with the same eigenvalue. Moreover, we generate a hierarchy of rogue wave solutions.


We start with a general periodic solution
\begin{equation}
  \label{seed}
  q=c e^{({\rm i}(ax+bt))},\quad
  b=a(-c^2+a),\quad a,c\in\mathbb{R}.
\end{equation}
Substituting \eqref{seed} into the spectral problem \eqref{lax}, we obtain the eigenfunction
$\Psi=\left(
\begin{array}{c}
f\\
g
\end{array}
\right)$
corresponding to the eigenvalue $\lambda$ via applying the method of separation of variables and superposition principle.


\begin{equation}\label{eigenfun}
\left(
\begin{array}{c}
f(x,t,\lambda)\\
g(x,t,\lambda)
\end{array}
\right)=\left(
\begin{array}{c}
D_1\omega^1_{11}(x,t,\lambda)+D_2\omega^2_{11}(x,t,\lambda)+D_1{\omega^1_{12}}^\ast(x,t,-\lambda^\ast)+D_2{\omega^2_{12}}^\ast(x,t,-\lambda^\ast)\\
D_1\omega^1_{12}(x,t,\lambda)+D_2\omega^2_{12}(x,t,\lambda)+D_1{\omega^1_{12}}^\ast(x,t,-\lambda^\ast)+D_2{\omega^2_{12}}^\ast(x,t,-\lambda^\ast)
\end{array}
\right).
\end{equation}


where
\begin{equation}\label{D1}
\left\{
\begin{aligned}
D_1&=1,\\
D_2&=1.
\end{aligned}
\right.
\end{equation}

\begin{equation}
\left(
\begin{array}{c}
\omega^1_{11}(x,t,\lambda)\\
\omega^1_{12}(x,t,\lambda)
\end{array}
\right)=\left(
\begin{array}{c}
\exp({c_1(x+2\lambda^2t+(-c^2+a)t)+\dfrac{1}{2}({\rm i}(ax+bt))})\\
\dfrac{{\rm i}a-2{\rm i}\lambda^2+2c_1}{2{\lambda}c}\exp({c_1(x+2\lambda^2t+bt)-\dfrac{1}{2}({\rm i}(ax+bt))})
\end{array}
\right),\nonumber
\end{equation}


\begin{equation}
\left(\begin{array}{c}
\omega^2_{11}(x,t,\lambda)\\
\omega^2_{12}(x,t,\lambda)
\end{array}
\right)=\left(\begin{array}{c}
\exp(-c_1(x+2\lambda^2t+(-c^2+a)t)+\dfrac{1}{2}({\rm i}(ax+bt)))\\
\dfrac{{\rm i}a-2{\rm i}\lambda^2-2c_1}{2{\lambda}c}\exp(-c_1(x+2\lambda^2t+bt)-\dfrac{1}{2}({\rm i}(ax+bt)))
\end{array}
\right),\nonumber
\end{equation}


\begin{equation}
\omega^1(x,t,\lambda)=\left(\begin{array}{c}
\omega^1_{11}(x,t,\lambda)\\
\omega^1_{12}(x,t,\lambda)\nonumber
\end{array}
\right),
\quad
\omega^2(x,t,\lambda)=\left(\begin{array}{c}
\omega^2_{11}(x,t,\lambda)\\
\omega^2_{12}(x,t,\lambda)\nonumber
\end{array}
\right),
\end{equation}
\begin{equation}
c_1=\dfrac{\sqrt{-a^2-4\lambda^4-4\lambda^2(c^2-a)}}{2}.\nonumber
\end{equation}


To get the rogue wave solutions, the value of $n$ in formula (\ref{q[n]}) must be even. When $n=2k$, we obtain the new expression of $q^{[n]}$.
\begin{proposition}
  Assuming $\lambda_1=\frac{1}{2}\sqrt {-{c}^{2}+2\,a}-\frac{1}{2}\,{\rm i}c$, $\lambda_2=-\lambda_1^*$, then $q^{[n]}$ given by following formula is the $k$-th rogue wave solution  for the DNLS equation.
  \begin{equation}\label{RW}
q^{[n]}=\frac{\delta_{11}^2}{\delta_{21}^2}q+2{\rm i}\frac{\delta_{11}\delta_{12}}{\delta_{21}^2};
\end{equation}


where
\begin{equation}
\delta_{11}=\begin{vmatrix}
g[1,n-1,1]&f[1,n-2,1]&g[1,n-3,1]&\cdots&g[1,1,1]&f[1,0,1]\\
g[2,n-1,1]&f[2,n-2,1]&g[2,n-3,1]&\cdots&g[2,1,1]&f[2,0,1]\\
g[1,n-1,2]&f[1,n-2,2]&g[1,n-3,2]&\cdots&g[1,1,2]&f[1,0,2]\\
g[2,n-1,2]&f[2,n-2,2]&g[2,n-3,2]&\cdots&g[2,1,2]&f[2,0,2]\\
\vdots&\vdots&\vdots&\vdots&\vdots&\vdots\\
g[1,n-1,k]&f[1,n-2,k]&g[1,n-3,k]&\cdots&g[1,1,k]&f[1,0,k]\\
g[2,n-1,k]&f[2,n-2,k]&g[2,n-3,k]&\cdots&g[2,1,k]&f[2,0,k]\nonumber
\end{vmatrix},
\end{equation}
\begin{equation}
\delta_{12}=\begin{vmatrix}
f[1,n,1]&f[1,n-2,1]&g[1,n-3,1]&\cdots&g[1,1,1]&f[1,0,1]\\
f[2,n,1]&f[2,n-2,1]&g[2,n-3,1]&\cdots&g[2,1,1]&f[2,0,1]\\
f[1,n,2]&f[1,n-2,2]&g[1,n-3,2]&\cdots&g[1,1,2]&f[1,0,2]\\
f[2,n,2]&f[2,n-2,2]&g[2,n-3,2]&\cdots&g[2,1,2]&f[2,0,2]\\
\vdots&\vdots&\vdots&\vdots&\vdots&\vdots\\
f[1,n,k]&f[1,n-2,k]&g[1,n-3,k]&\cdots&g[1,1,k]&f[1,0,k]\\
f[2,n,k]&f[2,n-2,k]&g[2,n-3,k]&\cdots&g[2,1,k]&f[2,0,k]\nonumber
\end{vmatrix},
\end{equation}
\begin{equation}
\delta_{21}=\begin{vmatrix}
f[1,n-1,1]&g[1,n-2,1]&f[1,n-3,1]&\cdots&f[1,1,1]&g[1,0,1]\\
f[2,n-1,1]&g[2,n-2,1]&f[2,n-3,1]&\cdots&f[2,1,1]&g[2,0,1]\\
f[1,n-1,2]&g[1,n-2,2]&f[1,n-3,2]&\cdots&f[1,1,2]&g[1,0,2]\\
f[2,n-1,2]&g[2,n-2,2]&f[2,n-3,2]&\cdots&f[2,1,2]&g[2,0,2]\\
\vdots&\vdots&\vdots&\vdots&\vdots&\vdots\\
f[1,n-1,k]&g[1,n-2,k]&f[1,n-3,k]&\cdots&f[1,1,k]&g[1,0,k]\\
f[2,n-1,k]&g[2,n-2,k]&f[2,n-3,k]&\cdots&f[2,1,k]&g[2,0,k]
\end{vmatrix}.\nonumber
\end{equation}
\end{proposition}


\begin{proof}
For the entries in the first column of $\Omega_{11}$ \eqref{q[n]},
\begin{equation*}
\begin{aligned}
\lambda_1^{n-1}g_1=&g[1,n-1,1]\epsilon,\\
\lambda_2^{n-1}g_2=&g[2,n-1,1]\epsilon,\\
\lambda_3^{n-1}g_3=&g[1,n-1,1]\epsilon+g[1,n-1,2]\epsilon^2,\\
\lambda_4^{n-1}g_4=&g[2,n-1,1]\epsilon+g[2,n-1,2]\epsilon^2,\\
&\vdots\\
\lambda_{n-1}^{n-1}g_{n-1}=&g[1,n-1,1]\epsilon+g[1,n-1,2]\epsilon^2+\cdots+g[1,n-1,k]\epsilon^{k},\\
\lambda_n^{n-1}g_n=&g[2,n-1,1]\epsilon+g[2,n-1,2]\epsilon^2+\cdots+g[2,n-1,k]\epsilon^{k}.\\
\end{aligned}
\end{equation*}
Taking the similar procedure to the other entries in $\Omega_{11}$, $\Omega_{12}$, and $\Omega_{21}$. Finally, the $q^{[n]}$ can be obtained through simple calculation.
\end{proof}


Here we point out that the formula \eqref{RW} is different with the result of \cite{Guo}.

Next, we present some special examples, which have different structures.

\begin{itemize}
  \item The first order rogue wave solution

  For $n=2$, we get the first order rogue wave solution according the above formulae
  \begin{equation}\label{RW1}
  q_{1_{rw}}^{[2]}=\frac{L_1^*L_2}{L_1^2}c{{\rm exp}^{{\rm i}a \left( x-t{c}^{2}+ta \right) }},
  \end{equation}
  where
  \begin{equation*}
  \begin{split}
  L_1=&e_1+{\rm i}e_2,\quad
  L_2=e_3+{\rm i}e_4,\\
  e_1=&-8\,{t}^{2}{c}^{2}{a}^{3}+12\,{t}^{2}{c}^{4}{a}^{2}-8\,x{c}^{2}t{a}^{2
  }+8\,x{c}^{4}ta-2\,a{x}^{2}{c}^{2}-6\,{t}^{2}{c}^{6}a-1,\\
  e_2=&4\,at{c}^{2}-6\,t{c}^{4}+2\,x{c}^{2},\\
  e_3=&8\,{t}^{2}{c}^{2}{a}^{3}+8\,x{c}^{2}t{a}^{2}-12\,{t}^{2}{c}^{4}{a}^{2}
       +2\,a{x}^{2}{c}^{2}-8\,x{c}^{4}ta+6\,{t}^{2}{c}^{6}a-3, \\
  e_4=&12\,at{c}^{2}-6\,t{c}^{4}+2\,x{c}^{2}.\\
  \end{split}
  \end{equation*}


  A direct analysis shows
  when $x\rightarrow \infty$, $t\rightarrow \infty$, $q^{[2]}\rightarrow c^2$, the maximum module of $\left|q^{[2]}\right|^2$ is equal to $9c^2$ and
  locates at the origin. A $1$-rogue wave with particular parameters is shown in Fig. \ref{fig.1rw}



  \item High order rogue wave solutions

  Generally, the expression of rogue wave becomes more complicated with increasing $n$ (an analytic expression of the second rogue wave solution is displayed in appendix A).
  Therefore,
  we use numerical simulations to discuss the high order rogue wave for convenient. We set $a=1$ and $c=1$  in the following.

  When $n=4$, we can obtain the second order rogue wave solution of the DNLS equation according to the formula (\ref{RW}).
  \begin{equation}\label{RW2}
  q_{2_{rw}}^{[4]}=\frac{L_1^*L_2}{L_1^2}{\rm exp}({\rm i}x),
  \end{equation}
  with
  \begin{equation*}
  \begin{split}
    L_1=&9+90\,{x}^{2}-12\,{x}^{4}+666\,{t}^{2}+180\,{t}^{4}+8\,{x}^{6}+8\,{t}^
        {6}-54\,{\rm i}x+24\,{\rm i}t{x}^{4}-216\,{x}^{2}{t}^{2}\\
        &-72\,xt+24\,{x}^{4}{t}^{2}+48\,{x}^{3}t+48\,x{t}^{3}+288\,{\rm i}{t}^{2}x+24\,{x}^{2}{t}^{4}
        -24\,{\rm i}{t}^{4}x-24\,{\rm i}{x}^{5}\\
        &-48\,{\rm i}{x}^{3}-48\,{\rm i}{t}^{2}{x}^{3}+198\,{\rm i}t+24\,{\rm i}{t}^{5}+48\,{\rm i}{t}^{3}{x}^{2}+336\,{\rm i}{t}^{3},\\
    L_2=&45-198\,{x}^{2}-60\,{x}^{4}-486\,{t}^{2}-60\,{t}^{4}-48\,{\rm i}{x}^{3}+528
         \,{\rm i}{t}^{3}+72\,{\rm i}{t}^{5}-414\,{\rm i}t+8\,{x}^{6}\\
         &+8\,{t}^{6}+72\,{\rm i}t{x}^{4}+144\,{\rm i}{t}^{3}{x}^{2}+24\,{\rm i}{t}^{4}x+48\,{\rm i}{t}^{2}{x}^{3}-576\,{\rm i}{t}^{2}x-
         288\,{\rm i}{x}^{2}t-90\,{\rm i}x\\&
         -504\,{x}^{2}{t}^{2}
         +504\,xt+24\,{x}^{4}{t}^{2}-
         144\,{x}^{3}t-144\,x{t}^{3}+24\,{\rm i}{x}^{5}+24\,{x}^{2}{t}^{4}.
  \end{split}
  \end{equation*}

  Besides, we succeed in reaching the $7$-order rogue wave by applying the above formula (\ref{RW}).
  Nevertheless, the analytical expression is too tedious, we omit it here. Their  dynamical evolutions are shown in Fig. \ref{fig.nrw}.
  From the figures, we find that the maximum height of the $k$-th order rogue wave
  is $(2k+1)^2$.
  As remarked in \cite{arxiv12093742,PRE84056611}, there
  are $\frac{k(k+1)}{2}-1$ local maxima on each side of the $x=0$ line for $k$-order rogue wave of the NLS equation.
  However, there are only $k$ small peak on the each side of the $t=0$ line in Fig. \ref{fig.nrw} for $k$-order rogue wave solution of the DNLS equation. We may make a conjecture here that the central peak of rogue wave of the DNLS  equation contain more energy than the NLS.
  \end{itemize}



 \section{{\bf The dynamics of rogue wave with parameters}}

 In above section, we assume that $D_1=1$ and $D_2=1$ \eqref{D1}. Actually, both $D_1$ and $D_2$ can be assumed as some new constants on the premise that \eqref{eigenfun} is the eigenfunction of spectral system \eqref{lax}. In this section, we set $D_1$ and $D_2$  as following:
 \begin{equation}\label{D3}
 \left\{
    \begin{aligned}
      D_1=&\exp\left(-{\rm i}c_1(S_0+S_1\epsilon+S_2\epsilon^2+S_3\epsilon^3+\cdots+S_{k-1}\epsilon^{k-1})\right),\\
      D_2=&\exp\left({\rm i}c_1(S_0+S_1\epsilon+S_2\epsilon^2+S_3\epsilon^3+\cdots+S_{k-1}\epsilon^{k-1})\right).
    \end{aligned}
    \right.
 \end{equation}
 Here $S_0,S_1,S_2,S_3,\cdots,S_{k-1}\in \mathbb{C}$ . Although the terms with nonzero
 orders of $\epsilon$ in eq.(\ref{D3}) vanish in the $\epsilon\rightarrow 0$ limit, their
 coefficients $S_i$($i=0,1,2,\cdots,k-1$) have a crucial effect on the structure of higher order rogue wave.
 Depending on these parameters,
 we can obtain a variety of solutions of the same order. Finding these relative positions in terms of $S_i$ ($i=0,1,2,\cdots,k-1$)
 is the subject of our analysis below.\\

 \noindent
 {\bf 3.1. Solutions with one parameter}\\

 In this subsection, we will discuss the dynamics of high-order rogue wave in detail. In the case that only one of the parameters is nonzero, four typical models are obtained: \emph{fundamental pattern}, \emph{triangular} structure, \emph{ring} structure, and \emph{modified-triangular} structure. Moreover, the \emph{modified-triangular} structure has never been found in other equations.

 \begin{itemize}
   \item Fundamental pattern

   For $n=2$, the first order rogue wave possesses only one parameter $S_0$, which is shown in Fig. \ref{fig.1rws0}. We observe that
   it is a translation of the solution in Fig. \ref{fig.1rw} relative to the origin. Actually, it can shift the $1$-order rogue wave solution to an arbitrary position on the ($x$,$t$)-plane, but it is trivial. So we omit this case for higher order solutions.\\


   \item Triangular structure

   In this subsection, we set $S_i=0$ except $S_1$. The resulting wave functions for orders $k=2,3,4,5,6,7$ are shown in Fig. \ref{fig.triangle}.  Remarkably, all higher order solutions display \emph{triangular} structure.
   We observe triangle with three peaks in Fig. \ref{fig.triangle}(a)(It has been obtained in Ref \cite{Guo}.), ten peaks in Fig. \ref{fig.triangle}(b), fifteen peaks in Fig. \ref{fig.triangle}(c), etc. All peaks within the triangle are first order rogue waves.
   So we can conclude that the \emph{triangular} structure of an order
   $k$ rogue wave solution is composed of $\frac{k(k+1)}{2}$ first order rogue waves, and it can be observed that successive $k$ rows possessing $k$, $k-1$, $k-2$, $\cdots$, $1$ peaks respectively.
   Evidently, the structure of the second order rogue wave is same as the
   result of the  NLS  equation,
   which is called triplet \cite{PLA3752782}.



   Another remarkable feature of these solutions is that, for the triangle of $k$-order($k>2$) rogue wave solutions, the outer triangle is composed
   of $3k-3$ first order rogue waves, and the inner triangle contains $\frac{k^2-5k+6}{2}$ first order rogue waves which is similar to the
   \emph{triangular} structure of $(k-3)$-th order rogue wave solution.
   For example, the $7$-th order rogue wave solution in Fig. \ref{fig.triangle}(f) is composed of $28$ first order rogue waves, $18$
   first order rogue waves locating on the outer shell,  and the inner is similar to the triangle of the forth order rogue wave
   containing $10$ first order rogue waves.


   \item Modified-triangular structure

   Actually, the inner triangle structure can form  a higher order rogue wave inversely  by changing the appearance of (\ref{D3}). For instance, when $k=5$, if we set

   \begin{equation}
   \label{eigenfun1}
   \left(
   \begin{array}{c}
   f(x,t,\lambda)\\
   g(x,t,\lambda)
   \end{array}
   \right)=\left(
   \begin{array}{c}
   D_1\omega^1_{11}(x,t,\lambda)+D_1\omega^2_{11}(x,t,\lambda)+D_2{\omega^1_{12}}^\ast(x,t,-\lambda^\ast)+D_2{\omega^2_{12}}^\ast(x,t,-\lambda^\ast)\\
   D_1\omega^1_{12}(x,t,\lambda)+D_1\omega^2_{12}(x,t,\lambda)+D_2{\omega^1_{12}}^\ast(x,t,-\lambda^\ast)+D_2{\omega^2_{12}}^\ast(x,t,-\lambda^\ast)
   \end{array}
   \right),\nonumber
   \end{equation}

   with

   \begin{equation}\nonumber
   \left\{
   \begin{aligned}
   D_1=&\exp(-{\rm i}c_1^2(S_0+S_1\epsilon+S_2\epsilon^2+S_3\epsilon^3+\cdots+S_{k-1}\epsilon^{k-1})),\\
   D_2=&\exp({\rm i}c_1^2(S_0+S_1\epsilon+S_2\epsilon^2+S_3\epsilon^3+\cdots+S_{k-1}\epsilon^{k-1})).
   \end{aligned}
   \right.
   \end{equation}

   we will get a triangular structure with an second order rogue wave located in the center. It is remarkable that this structure has never been given before in nonlinear science, which is called \emph{modified-triangular} structure. One with special parameter is shown in Fig. \ref{fig.newtriangle}.

   \item Ring structure

   If we assume $S_i=0$ except $S_{k-1}$, we can get \emph{ring} structures, which are shown in Fig. \ref{fig.circle}. They possess $1$-order rogue waves and higher order rogue waves. Peaks locating on the outer shell of the ring are all first order rogue waves, and locating in the center of the ring are higher order rogue waves. Besides, the number of the first order rogue wave and the order of the inner higher order rogue wave increase according to the order of rogue wave. From those figures, we can conclude that there are $2k-1$ first order rogue waves locating on the outer shell of \emph{ring} structure of the $k$-th order rogue wave solution, and a WANDT of order $k-2$ locates in the center of ring. Notably, this structure has never been displayed for the DNLS equation.



   \end{itemize}

\noindent
{\bf 3.2. Solutions with more than one parameter}\\

 Generally, there are $k-1$ free parameters for $k$-order rogue wave solution. As we discussed in previous subsection, the four basic models are depending on particular parameters. If there are two or more parameters which are non-zero, new models will be obtained. Moreover, higher order WANDT locating in the center of ring structure can be split into lower order waves.\\


\begin{itemize}

  \item Ring-triangle

  When $k=4$ with parameter $S_3\neq0$, we have gotten a \emph{ring} structure with $7$ $1$-order rogue wave solution located on the outer shell,
  and a $2$-order rogue wave locating in the center in above section. Further more, if the parameter $S_1$ is also non-zero, the central higher order peak is split
  into a \emph{triangular} structure and the outer shell remains the same.  When $k=5,6,7$, the similar structure is also displayed.
  These phenomenons are
  displayed in Fig. \ref{fig.circletriangle}. Therefore, we are able to conclude that the central higher order rogue wave will be split into a \emph{triangular} structure if $S_{k-1}\gg0$ and $S_1\neq0$ for $k$-order WANDT.\\

  \item Multi-ring

  Similarly, the central higher order WANDT in \emph{ring} structure can also be split into \emph{ring} structure. For instance, when $k=5$, we have observe a \emph{ring} structure
  in Fig. \ref{fig.circle}(b). In this case, if we set $S_2\neq0$, the inner $3$-order rogue wave can be split into a \emph{ring} model.  Its dynamics are shown in Fig. \ref{fig.5rwmulticircle}. When $k=6$, if we assume $S_5=1\times10^8$ and $S_4=1\times10^6$, the inner higher order rogue wave is split into a \emph{ring} structure with a second order rogue wave locating in the center. Its evolution is shown in Fig. \ref{fig.6rwmulticircle1}. Indeed, we can continuing decomposing the inner structure with the help of another parameters. A new \emph{multi-ring} model of the six order is displayed in Fig.\ref{fig.6rwmulticircle2}. From these figures, we find that both the outer shell and the middle shell are circular,  and the inner shell is triangular (or circular).  Naturally, the $7$-order rogue wave solution possesses the same character except for the difference that
  the central peak can be split into both a triangle model and a ring pattern. They are shown in Fig. \ref{7rwmulticircle1} and Fig. \ref{7rwmulticircle2}.
  \end{itemize}

\section{{\bf Conclusions}}

In this paper, we generate the formulae of higher order positon solution in  proposition 1 and higher order rogue wave solution in proposition 2 for the DNLS equation at the same eigenvalue with the method of Taylor expansion and limit technique. By applying these formulae, we get positon solutions, rational traveling solutions and rogue
wave solutions. These formulae are given in terms of determinants  explicitly. Remarkably, the formula for rogue wave solutions
is really effective in achieving the analytic expression and computer simulation of $k$-order
rogue wave. Further more, we give rise to solutions with different structures are obtained by adjusting the free parameters $D_1$ and $D_2$ in our formula. With the help of these parameters, we study the dynamics of higher order rogue wave solutions. Overall, there are four basic models for higher order rogue wave. By choosing proper parameters, combination structures can be obtained. For example,  \emph{fundamental pattern}, \emph{triangular} structure, \emph{ring} structure,\emph{modified-triangular} structure, \emph{ring-triangle} structure, and \emph{multi-ring} structure. The last three models have never been given before.

In the last part of this paper, we make a classification of higher order rogue wave of the DNLS equation. We found out that some basis structures (\emph{fundamental pattern}, \emph{triangular} structure, \emph{ring} structure) also appear in other equations such as  NLS. But the \emph{modified-triangular} structure is unique to the DNLS equation.

Our results give an essential understanding of the relation of shift parameters with
relative positions, which will be useful in other integrable equations such as Hirota equation,
Gerdjikov-Ivanov equation, the Davey-Stewartson equation, and so on.

\mbox{\vspace{1cm}}

{\bf Acknowledgments}

{\noindent This work is supported by the NSF of China under Grant No.11271210, No.10971109 and K. C. Wong
Magna Fund in Ningbo University. Jingsong He is also supported by Natural Science Foundation of
Ningbo under Grant No. 2011A610179. We want to thank Prof. Yishen Li (USTC, Hefei, China) for
his long time support and useful suggestions.}

\appendix
\section{The expression of the second order rogue wave solution}
Here we present the expression of the second order rogue wave solution
\begin{equation}
  q_{2_{rw}}=-\frac{L_1^*L_2}{L_1^2}{\rm exp}({{\rm i} \left( ax+ \left( -{c}^{2}+a \right) at \right) }),
\end{equation}
where
\begin{equation}\nonumber
  \begin{split}
    L_1=&e_1+{\rm i}e_2,\qquad
    L_2=e_3+{\rm i}e_4,\\
    e_1=&-9+72\,a{c}^{4}xt-216\,{a}^{2}{c}^{2}xt-108\,{a}^{2}{c}^{4}{t}^{2}-18
         \,a{c}^{6}{t}^{2}+216\,{c}^{6}xt-54\,a{c}^{2}{x}^{2}-216\,{a}^{3}{c}^{2}{t}^{2}\\
       &+1088\,{a}^{3}{c}^{12}{x}^{3}{t}^{3}-1536\,{a}^{8}{c}^{6}x{t}^{5}
        +6528\,{a}^{6}{c}^{8}{x}^{2}{t}^{4}-3264\,{a}^{4}{c}^{10}{x}^{3}{t}^{3}
        -1920\,{a}^{7}{c}^{6}{x}^{2}{t}^{4}\\
        &+3456\,{a}^{5}{c}^{8}{x}^{3}{t}^{3}-456\,{a}^{3}{c}^{10}{x}^{4}{t}^{2}+648\,a{c}^{14}{t}^{4}
        -1280\,{a}^{6}{c}^{6}{x}^{3}{t}^{3}+912\,{a}^{4}{c}^{8}{x}^{4}{t}^{2}-2268\,{a
        }^{2}{c}^{12}{t}^{4}\\
        &-480\,{a}^{5}{c}^{6}{x}^{4}{t}^{2}+4272\,{a}^{3}{c
     }^{10}{t}^{4}+96\,{a}^{3}{c}^{8}{x}^{5}t-1296\,a{c}^{12}x{t}^{3}-4368
     \,{a}^{4}{c}^{8}{t}^{4}-96\,{a}^{4}{c}^{6}{x}^{5}t\\
  &+3456\,{a}^{2}{c}^{
    10}x{t}^{3}+1728\,{a}^{5}{c}^{6}{t}^{4}-4320\,{a}^{3}{c}^{8}x{t}^{3}-8
    \,{a}^{3}{c}^{6}{x}^{6}+864\,a{c}^{10}{x}^{2}{t}^{2}-216\,{a}^{3}{c}^{
    18}{t}^{6}\\
  &+1296\,{a}^{4}{c}^{16}{t}^{6}-3456\,{a}^{5}{c}^{14}{t}^{6}+
    864\,{a}^{3}{c}^{16}x{t}^{5}+5184\,{a}^{6}{c}^{12}{t}^{6}-4320\,{a}^{4
    }{c}^{14}x{t}^{5}-4608\,{a}^{7}{c}^{10}{t}^{6}\\
  &+9216\,{a}^{5}{c}^{12}x{
   t}^{5}-1368\,{a}^{3}{c}^{14}{x}^{2}{t}^{4}+2304\,{a}^{8}{c}^{8}{t}^{6}
   -10368\,{a}^{6}{c}^{10}x{t}^{5}+5472\,{a}^{4}{c}^{12}{x}^{2}{t}^{4}-
   512\,{a}^{9}{c}^{6}{t}^{6}\\
 &+6144\,{a}^{7}{c}^{8}x{t}^{5}-8736\,{a}^{5}{
   c}^{10}{x}^{2}{t}^{4}-96\,{a}^{3}{c}^{4}{x}^{3}t+24\,a{c}^{6}{x}^{4}-
   12\,{a}^{2}{c}^{4}{x}^{4}-36\,{c}^{4}{x}^{2}-324\,{c}^{8}{t}^{2}\\
 &-192\,
   {a}^{6}{c}^{4}{t}^{4}+2496\,{a}^{4}{c}^{6}x{t}^{3}-1656\,{a}^{2}{c}^{8
   }{x}^{2}{t}^{2}-384\,{a}^{5}{c}^{4}x{t}^{3}+1296\,{a}^{3}{c}^{6}{x}^{2
   }{t}^{2}-240\,a{c}^{8}{x}^{3}t\\
   &-288\,{a}^{4}{c}^{4}{x}^{2}{t}^{2}+288\,
     {a}^{2}{c}^{6}{x}^{3}t,\\
e_2=&6\,{c}^{2} ( 9\,x+12\,{c}^{4}{x}^{2}ta-24\,{c}^{4}x{t}^{2}{a}^{2}
     +48\,x{t}^{2}{c}^{2}{a}^{3}+24\,{x}^{2}t{a}^{2}{c}^{2}-180\,{c}^{6}{t}
      ^{2}xa+1776\,{c}^{8}{a}^{4}{t}^{4}x\\
      &-360\,{c}^{10}{a}^{2}{t}^{3}{x}^{2}
        -51\,t{c}^{2}+18\,ta+4\,{c}^{4}{a}^{2}{x}^{5}+108\,{c}^{8}{t}^{2}x-
            1008\,{c}^{10}{a}^{4}{t}^{5}-240\,{c}^{6}{t}^{3}{a}^{2}-\\
     &108\,{c}^{14}{
a}^{2}{t}^{5}-576\,{c}^{6}{a}^{6}{t}^{5}-64\,{c}^{4}{t}^{3}{a}^{3}+128
\,{c}^{4}{a}^{7}{t}^{5}+1056\,{c}^{8}{a}^{5}{t}^{5}+32\,{t}^{3}{c}^{2}
{a}^{4}+324\,{c}^{8}{t}^{3}a\\
&-36\,{c}^{6}t{x}^{2}+504\,{c}^{12}{a}^{3}{
t}^{5}+40\,{c}^{4}{a}^{3}t{x}^{4}+324\,{c}^{12}{a}^{2}{t}^{4}x-1200\,{
c}^{10}{a}^{3}{t}^{4}x-336\,{c}^{6}{a}^{3}{t}^{2}{x}^{3}\\
&+160\,{c}^{4}{
a}^{4}{t}^{2}{x}^{3}-44\,{c}^{6}{a}^{2}t{x}^{4}+4\,{c}^{4}{x}^{3}-108
\,{c}^{10}{t}^{3}-960\,{c}^{6}{a}^{4}{t}^{3}{x}^{2}+184\,{c}^{8}{a}^{2
}{t}^{2}{x}^{3}\\
&+320\,{c}^{4}{a}^{5}{t}^{3}{x}^{2}+4\,{x}^{3}{c}^{2}a-
1216\,{c}^{6}{a}^{5}{t}^{4}x+320\,{c}^{4}{a}^{6}{t}^{4}x+992\,{c}^{8}{
a}^{3}{t}^{3}{x}^{2} ),\\
e_3=&(e_5-e_1)c,\\
e_4=&6\,{c}^{3} ( -15\,x+108\,{c}^{4}{x}^{2}ta+456\,{c}^{4}x{t}^{2}{a}
^{2}-336\,x{t}^{2}{c}^{2}{a}^{3}-120\,{x}^{2}t{a}^{2}{c}^{2}-324\,{c}^
{6}{t}^{2}xa+2928\,{c}^{8}{a}^{4}{t}^{4}x\\
&-360\,{c}^{10}{a}^{2}{t}^{3}{
x}^{2}+21\,t{c}^{2}-90\,ta+4\,{c}^{4}{a}^{2}{x}^{5}+108\,{c}^{8}{t}^{2
}x-1584\,{c}^{10}{a}^{4}{t}^{5}-384\,{c}^{6}{t}^{3}{a}^{2}-108\,{c}^{
14}{a}^{2}{t}^{5}\\
&-1344\,{c}^{6}{a}^{6}{t}^{5}+544\,{c}^{4}{t}^{3}{a}^{
3}+384\,{c}^{4}{a}^{7}{t}^{5}+2016\,{c}^{8}{a}^{5}{t}^{5}-288\,{t}^{3}
{c}^{2}{a}^{4}+324\,{c}^{8}{t}^{3}a-36\,{c}^{6}t{x}^{2}\\
&+648\,{c}^{12}{
a}^{3}{t}^{5}+56\,{c}^{4}{a}^{3}t{x}^{4}+324\,{c}^{12}{a}^{2}{t}^{4}x-
1584\,{c}^{10}{a}^{3}{t}^{4}x-464\,{c}^{6}{a}^{3}{t}^{2}{x}^{3}+288\,{
c}^{4}{a}^{4}{t}^{2}{x}^{3}\\
&-44\,{c}^{6}{a}^{2}t{x}^{4}
+4\,{c}^{4}{x}^{
3}-108\,{c}^{10}{t}^{3}-1664\,{c}^{6}{a}^{4}{t}^{3}{x}^{2}+184\,{c}^{8
}{a}^{2}{t}^{2}{x}^{3}+704\,{c}^{4}{a}^{5}{t}^{3}{x}^{2}-12\,{x}^{3}{c
}^{2}a\\
&-2496\,{c}^{6}{a}^{5}{t}^{4}x+832\,{c}^{4}{a}^{6}{t}^{4}x+1344\,
{c}^{8}{a}^{3}{t}^{3}{x}^{2}),\\
e_5=&36+288\,a{c}^{4}xt-576\,{a}^{2}{c}^{2}xt+288\,{a}^{2}{c}^{4}{t}^{2}+
432\,a{c}^{6}{t}^{2}+864\,{c}^{6}xt-144\,a{c}^{2}{x}^{2}-576\,{a}^{3}{
c}^{2}{t}^{2}\\
&+1296\,{a}^{2}{c}^{12}{t}^{4}-4032\,{a}^{3}{c}^{10}{t}^{4
}+4032\,{a}^{4}{c}^{8}{t}^{4}-1728\,{a}^{2}{c}^{10}x{t}^{3}-768\,{a}^{
5}{c}^{6}{t}^{4}+3072\,{a}^{3}{c}^{8}x{t}^{3}\\
&-384\,{a}^{3}{c}^{4}{x}^{
3}t-48\,{a}^{2}{c}^{4}{x}^{4}-144\,{c}^{4}{x}^{2}-1296\,{c}^{8}{t}^{2}
-768\,{a}^{6}{c}^{4}{t}^{4}+288\,{a}^{2}{c}^{8}{x}^{2}{t}^{2}-1536\,{a
}^{5}{c}^{4}x{t}^{3}\\
&+576\,{a}^{3}{c}^{6}{x}^{2}{t}^{2}-1152\,{a}^{4}{c
}^{4}{x}^{2}{t}^{2}+192\,{a}^{2}{c}^{6}{x}^{3}t.
\end{split}
\end{equation}

With the help the formula \eqref{RW}, we can also obtain the expression of $k$-th (k=3,4,5,6,7) order  rogue wave. Since they are too complicated to write down, we omit them.


\newpage
\begin{figure}[!ht]
  \centering
  \subfigure[]{\label{Fig.sub.1}\includegraphics[height=4cm,width=4cm]{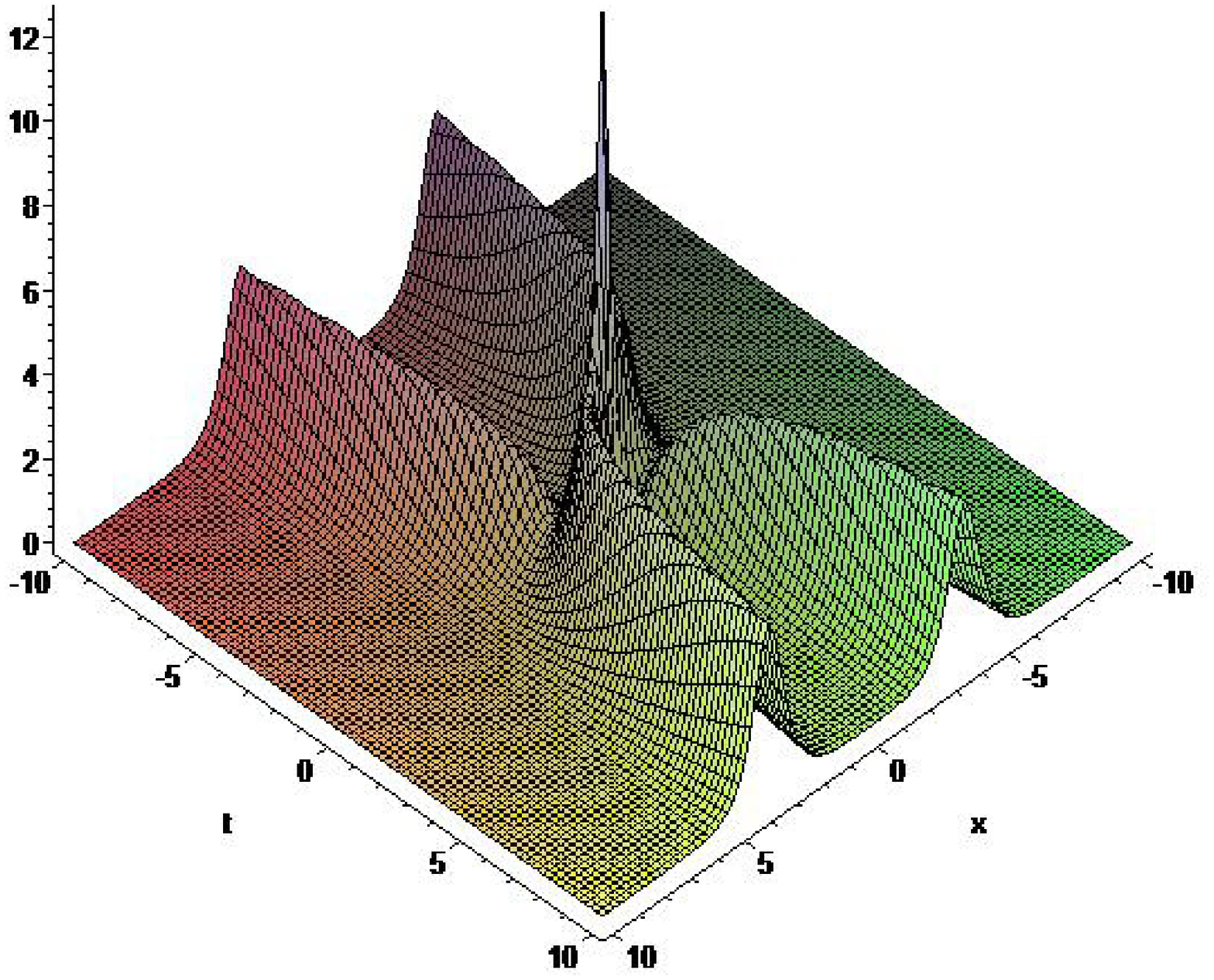}}
  \qquad
  \subfigure[]{\label{Fig.sub.2}\includegraphics[height=4cm,width=4cm]{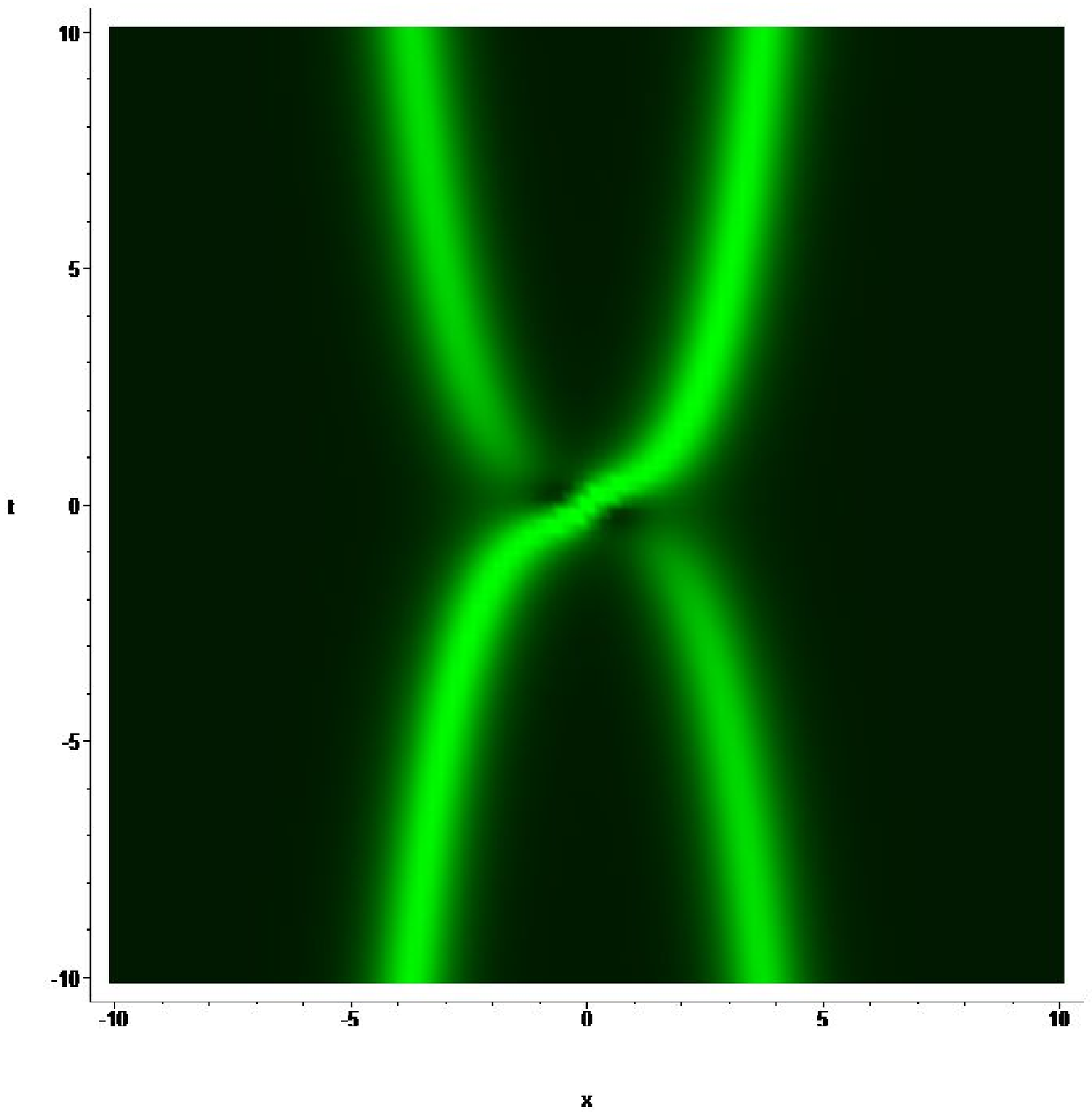}}
  \caption{(Color online) The dynamics of positon solution  on the (x, t) plane with $\alpha_1=0.5$, $\beta_1=0.5$.}
  \label{positon}
\end{figure}

\begin{figure}[!ht]
  \subfigure[]{\label{Fig.sub.3}\includegraphics[height=4cm,width=4cm]{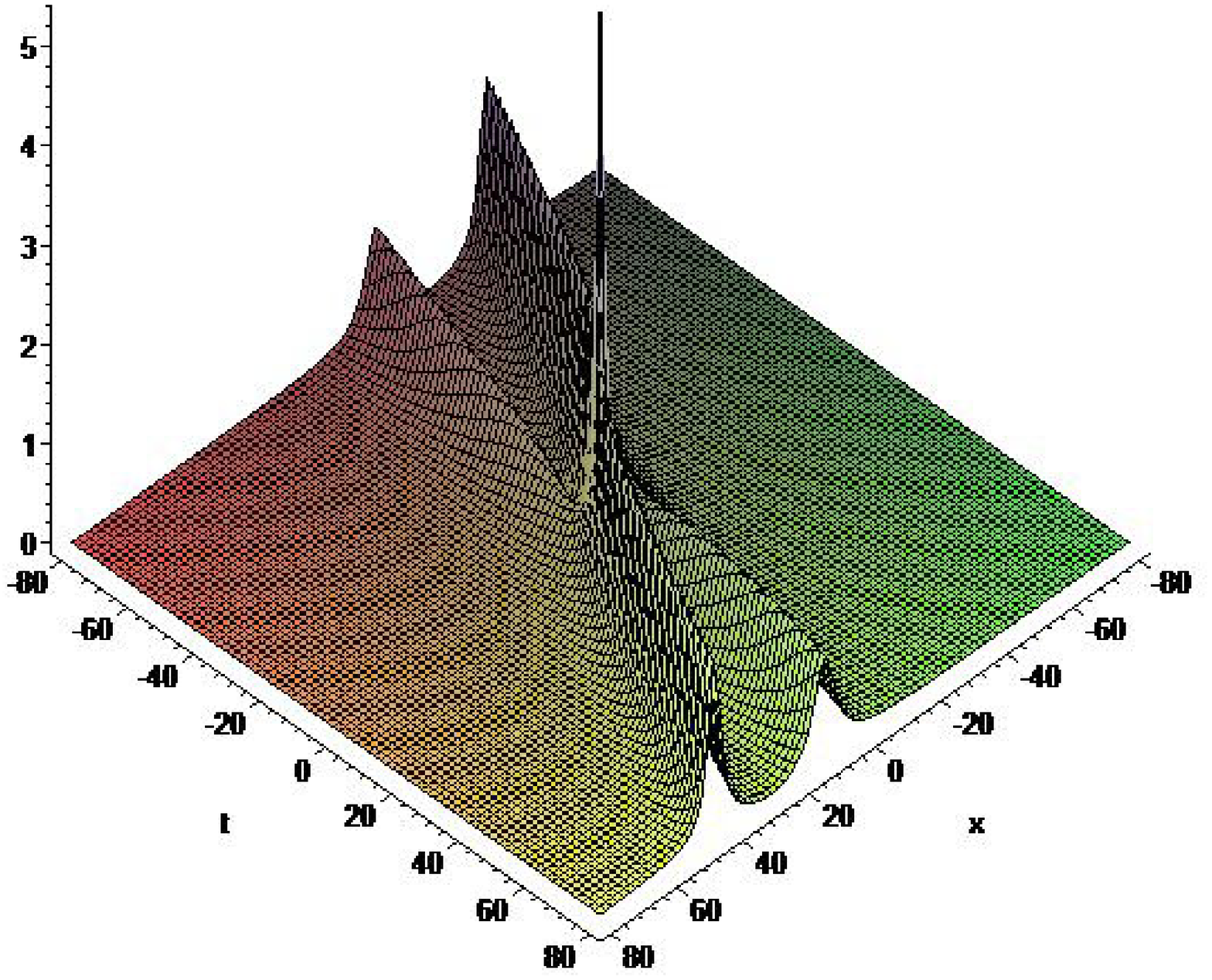}}
  \qquad
  \subfigure[]{\label{Fig.sub.4}\includegraphics[height=4cm,width=4cm]{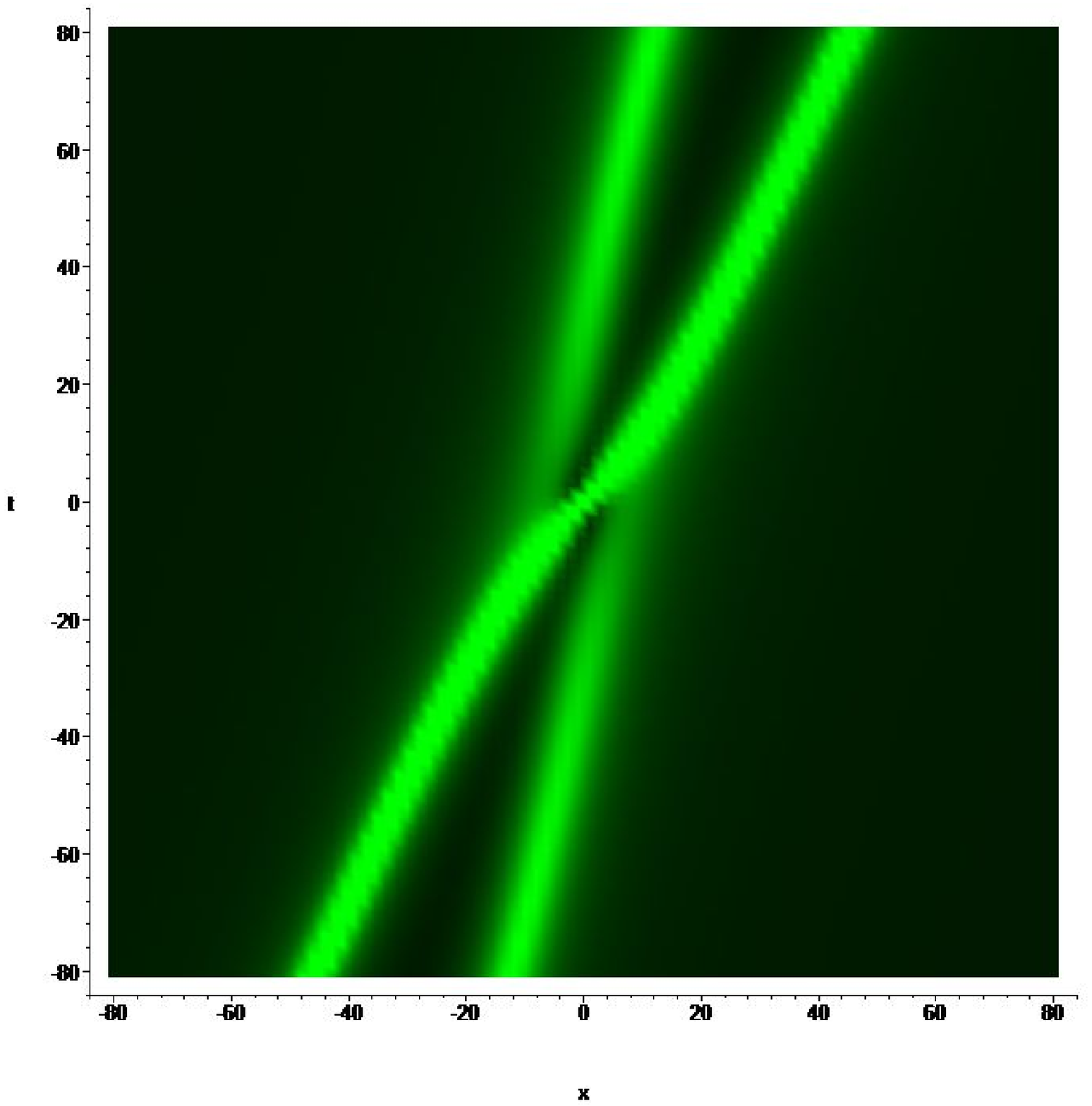}}
  \caption{(Color online) The dynamics of $2$-rd rational traveling solution on the (x, t) plane with $\beta_1=0.3$.}
  \label{ration}
\end{figure}


\begin{figure}[!htbp]
  \centering
  \includegraphics[width=4cm,height=4cm]{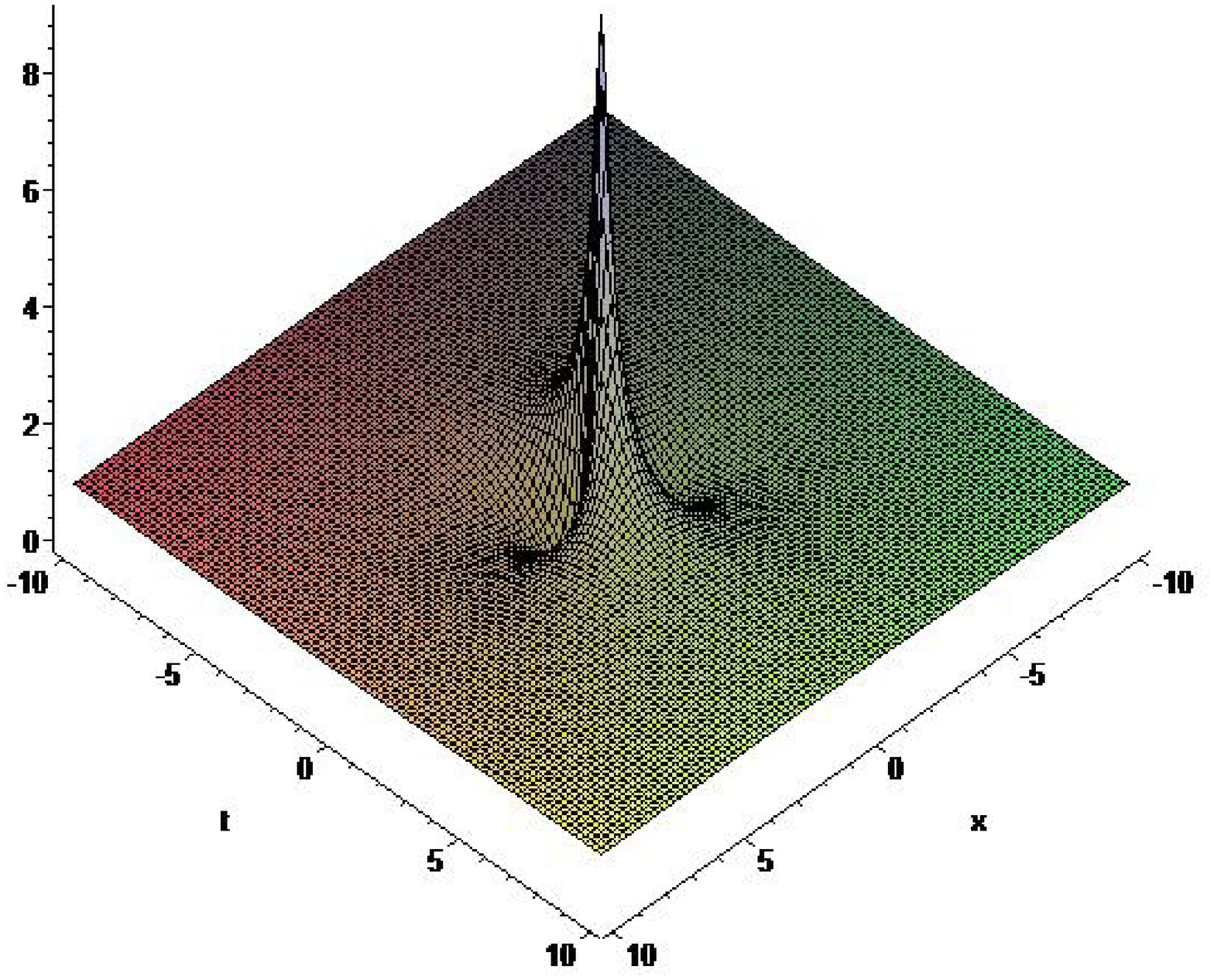}
  \qquad
  \includegraphics[width=4cm,height=4cm]{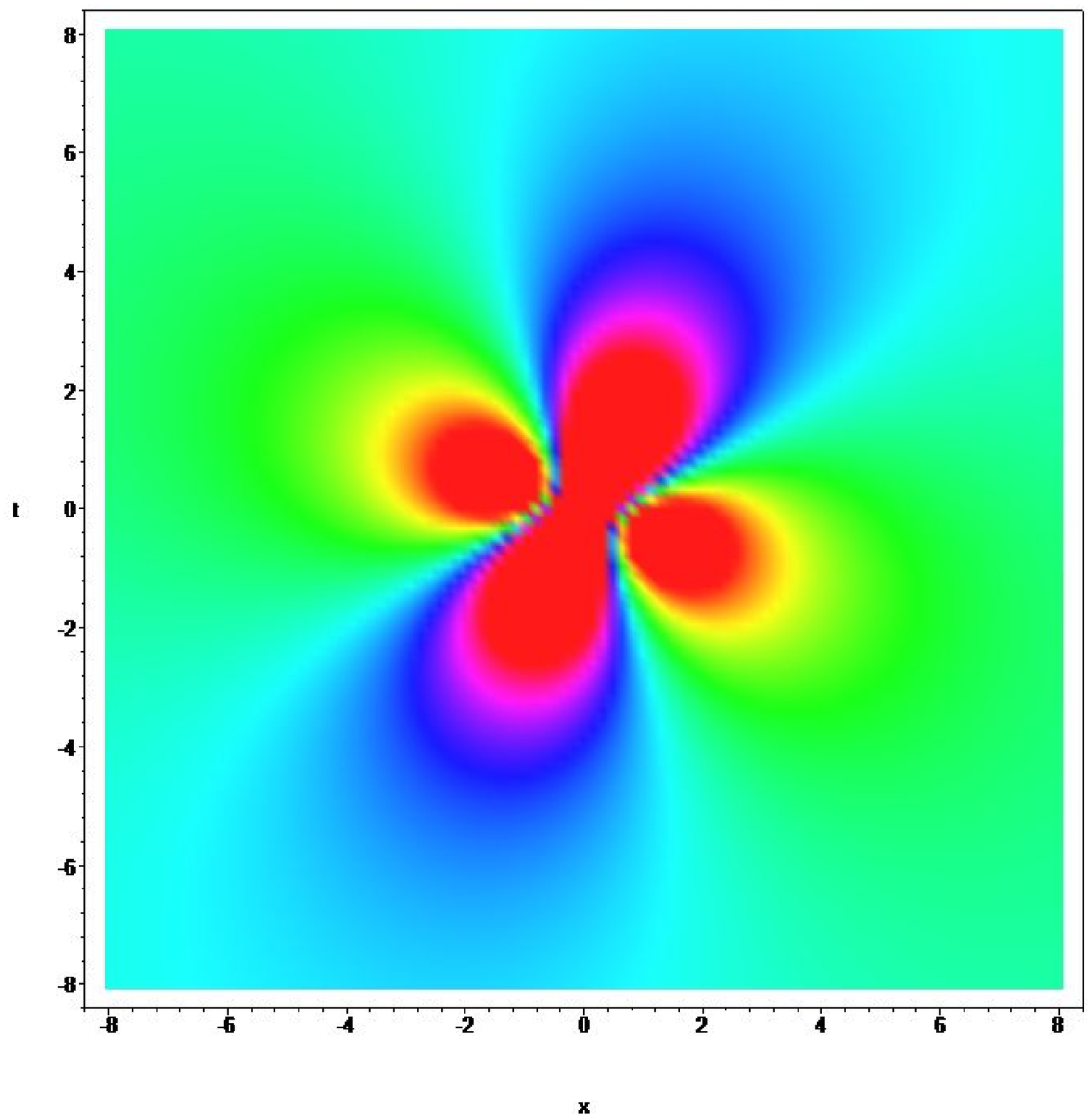}
  \caption{1-rogue wave solution with $a=1$, $c=1$.}\label{fig.1rw}
\end{figure}

\begin{figure}[!ht]
\centering
\subfigure[2-rogue wave]{\includegraphics[height=4cm,width=4cm]{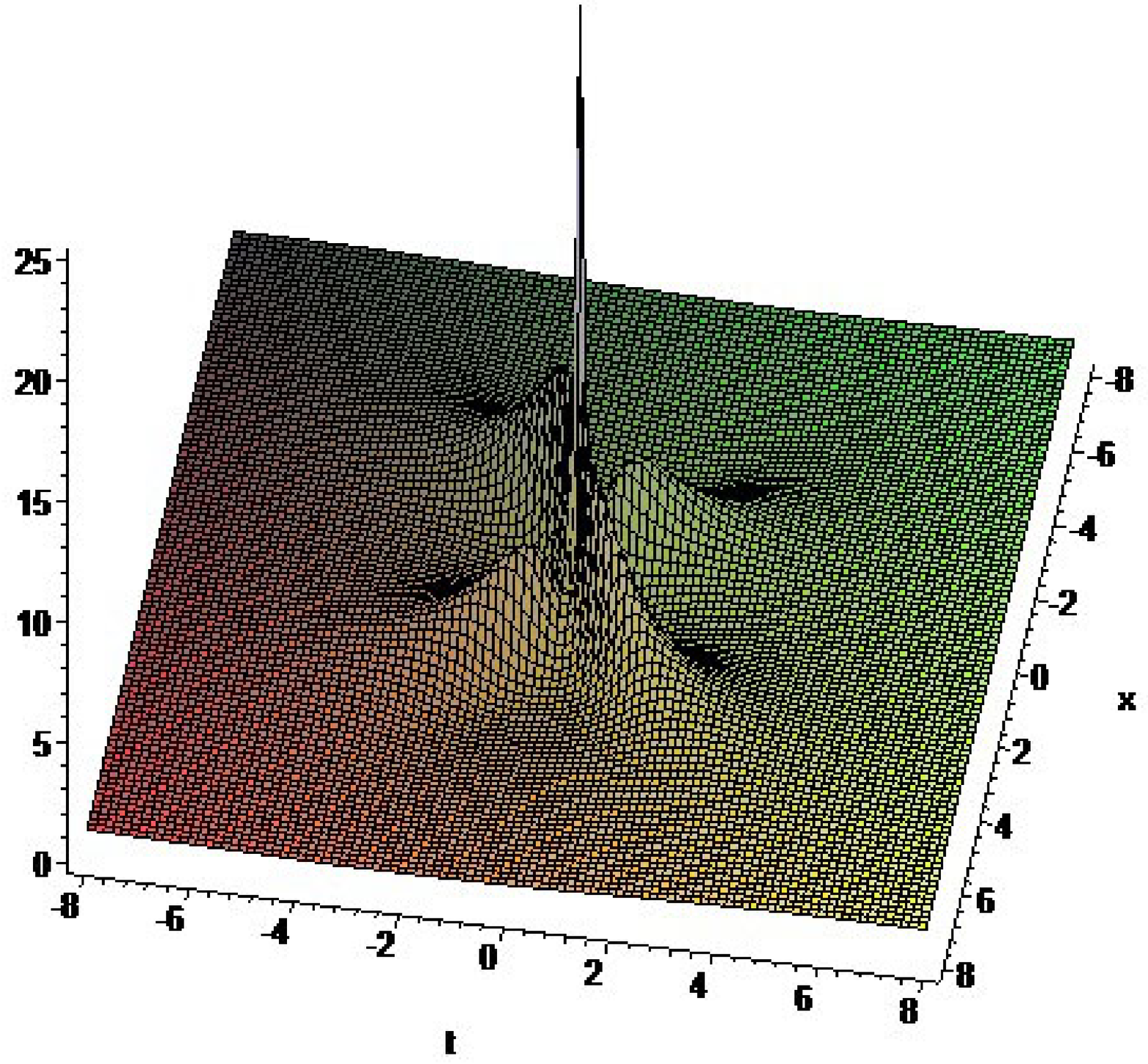}}
\qquad
\subfigure[3-rogue wave]{\includegraphics[height=4cm,width=4cm]{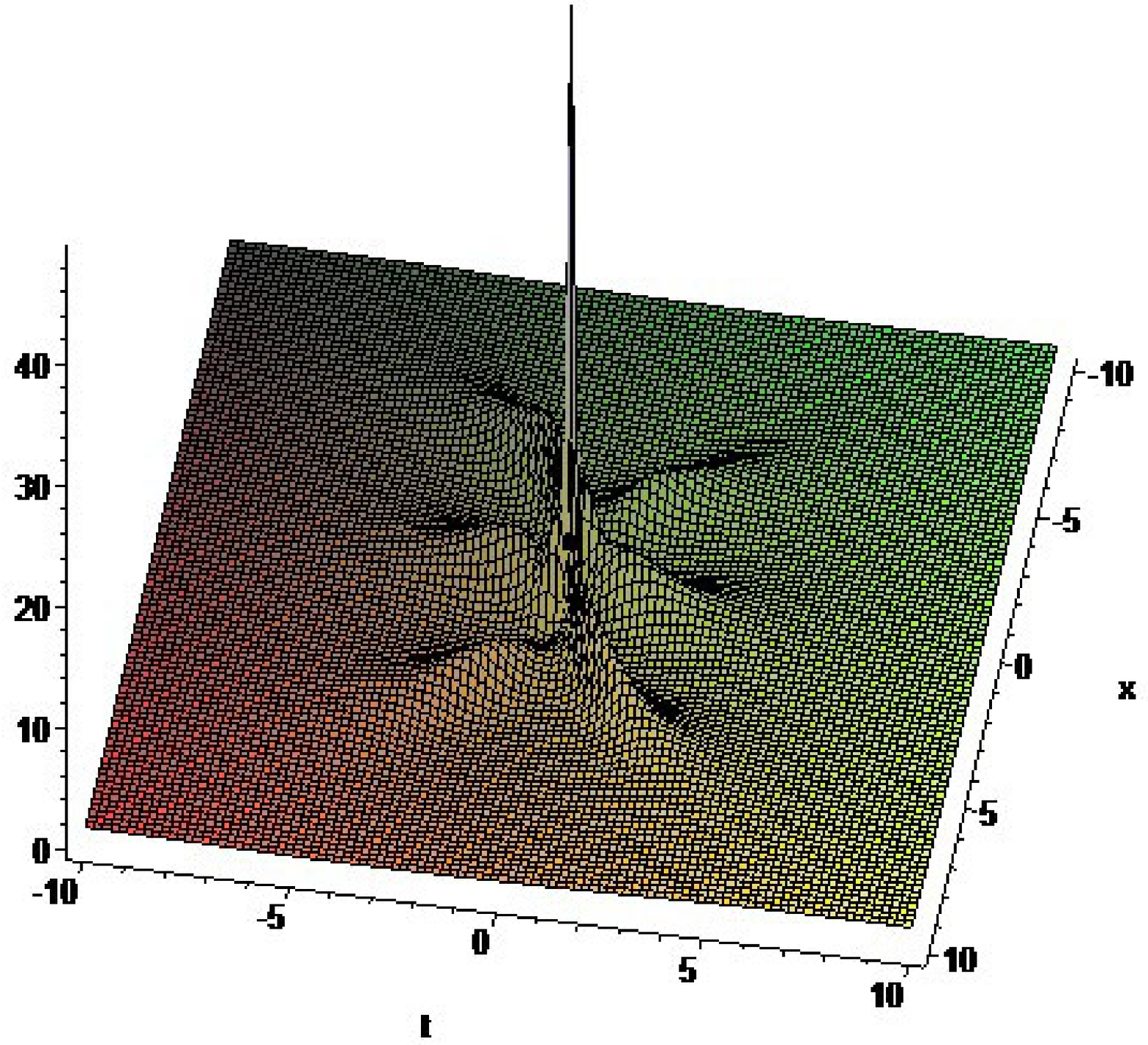}}
\qquad
\subfigure[4-rogue wave]{\includegraphics[height=4cm,width=4cm]{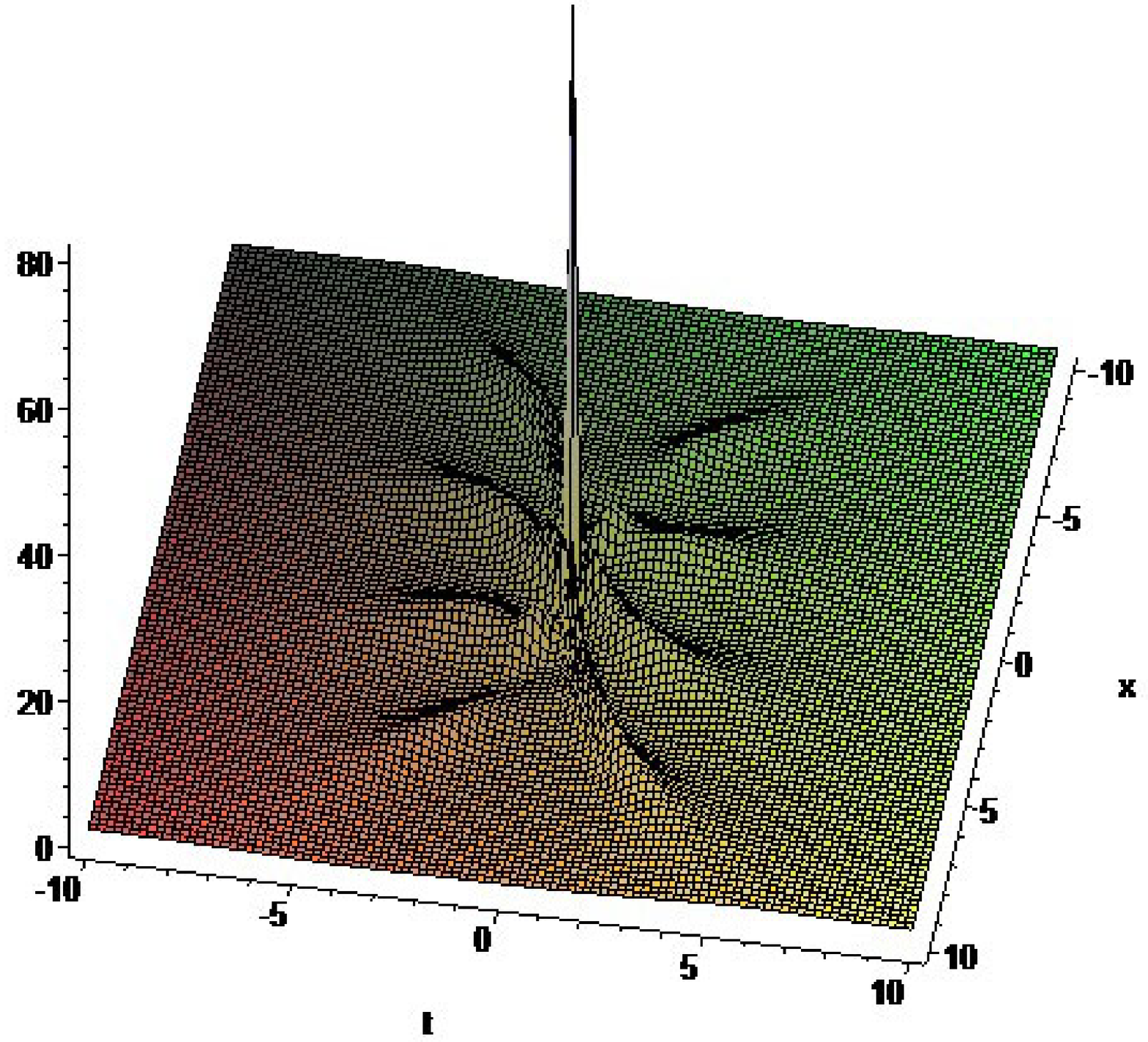}}
\qquad
\subfigure[5-rogue wave]{\includegraphics[height=4cm,width=4cm]{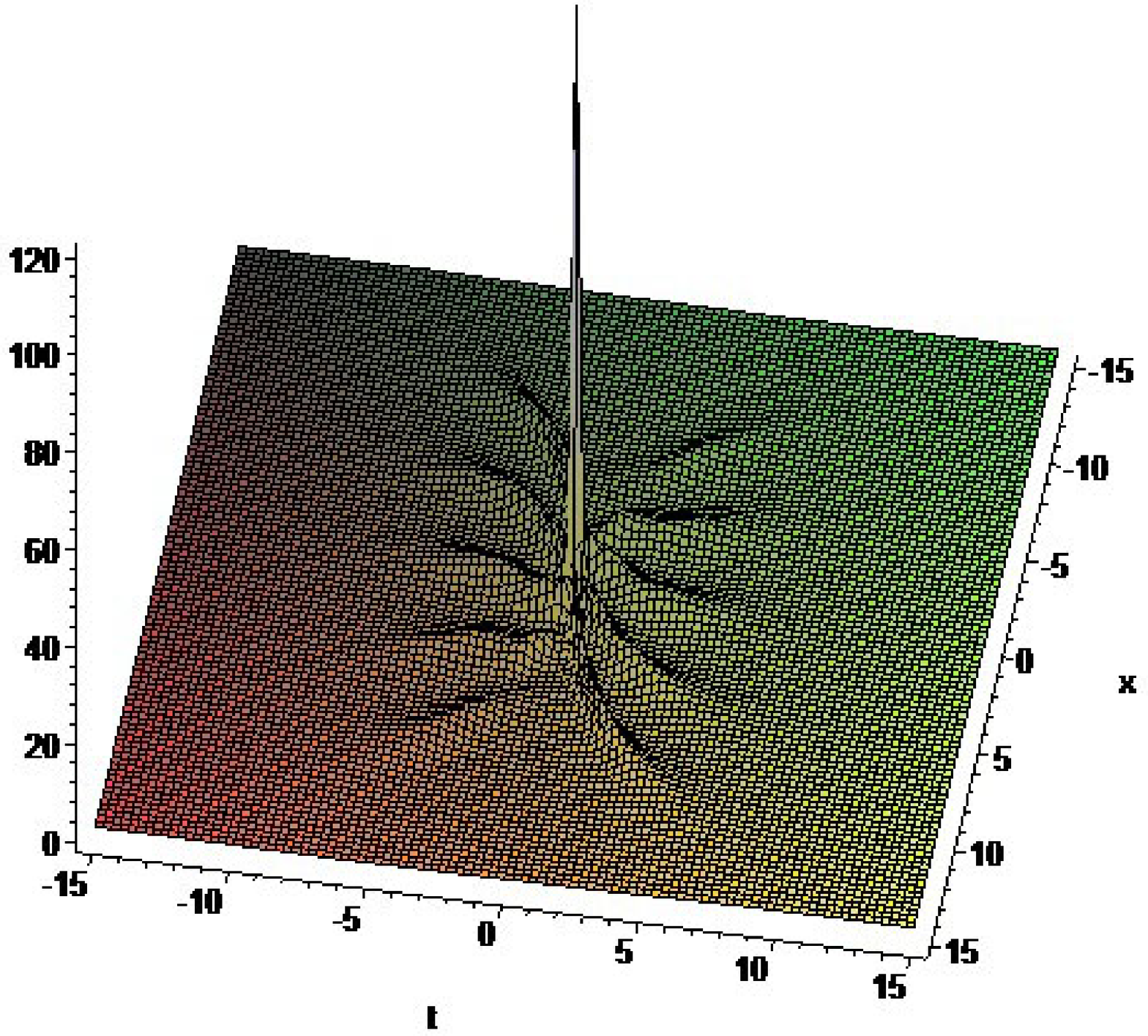}}
\qquad
\subfigure[6-rogue wave]{\includegraphics[height=4cm,width=4cm]{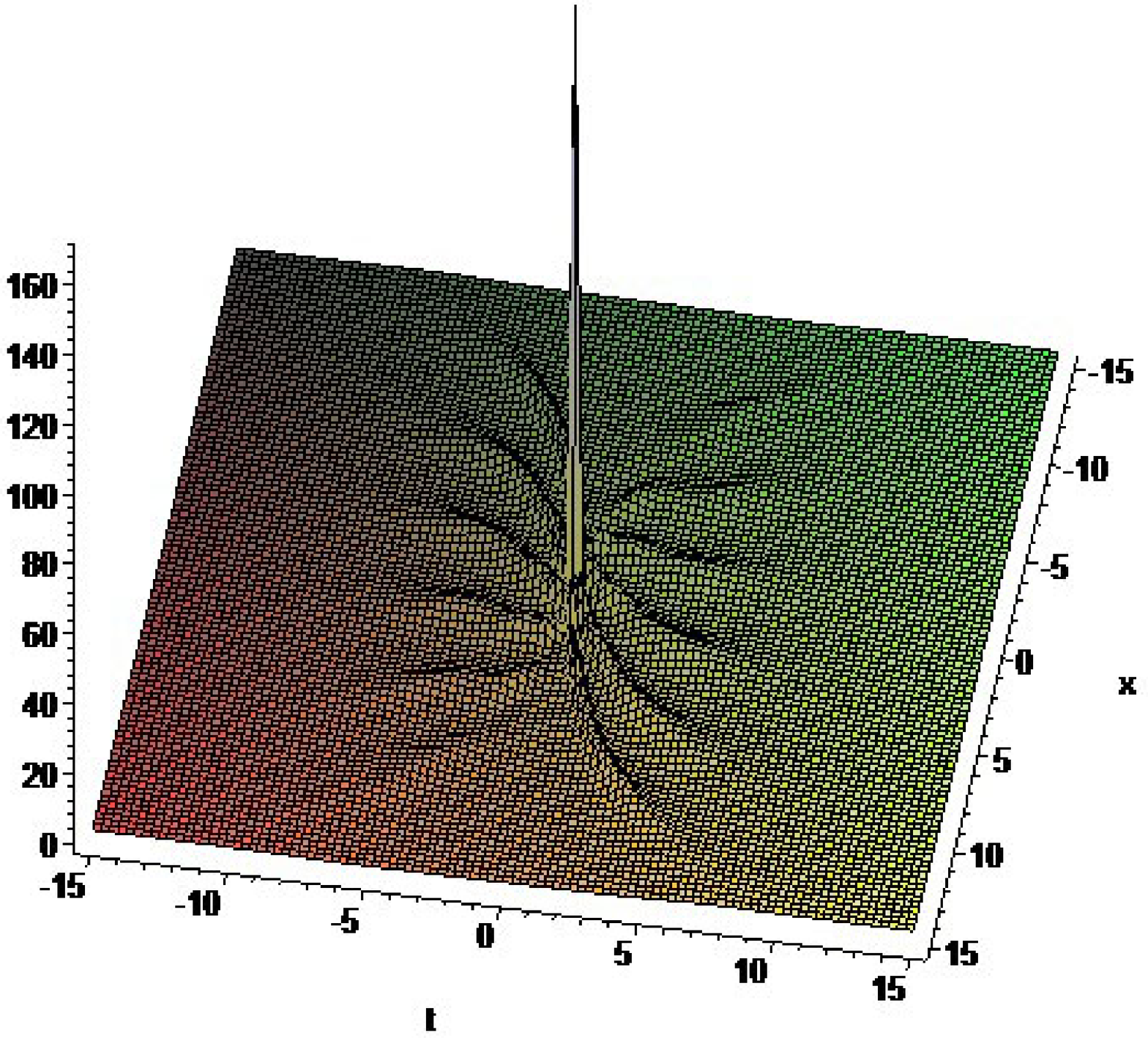}}
\qquad
\subfigure[7-rogue wave]{\includegraphics[height=4cm,width=4cm]{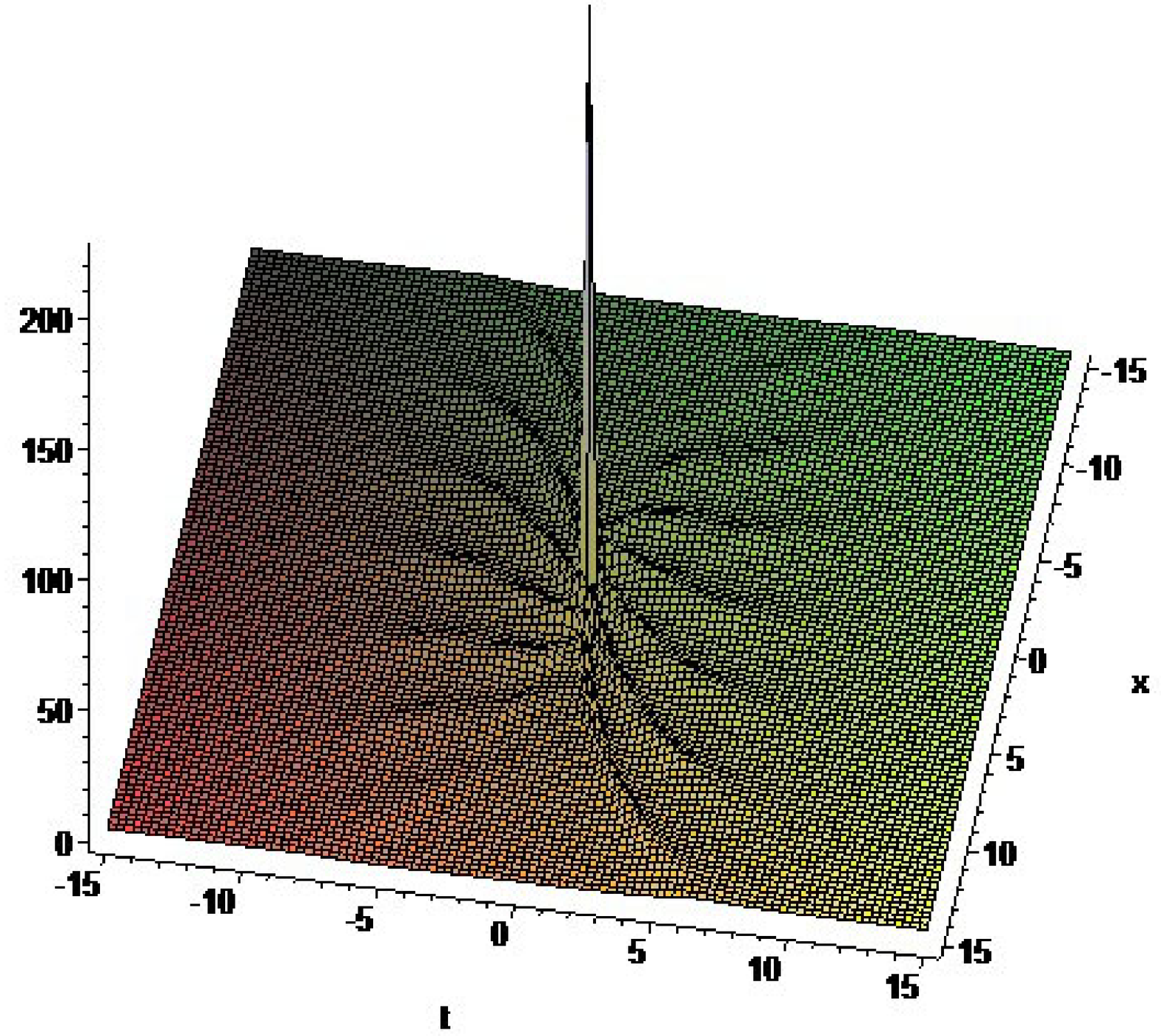}}
\caption{(Color online)The dynamics of higher order rogue wave solution.}\label{fig.nrw}
\end{figure}


\begin{figure}[!htp]
   \centering
   \subfigure[]{\includegraphics[height=4cm,width=4cm]{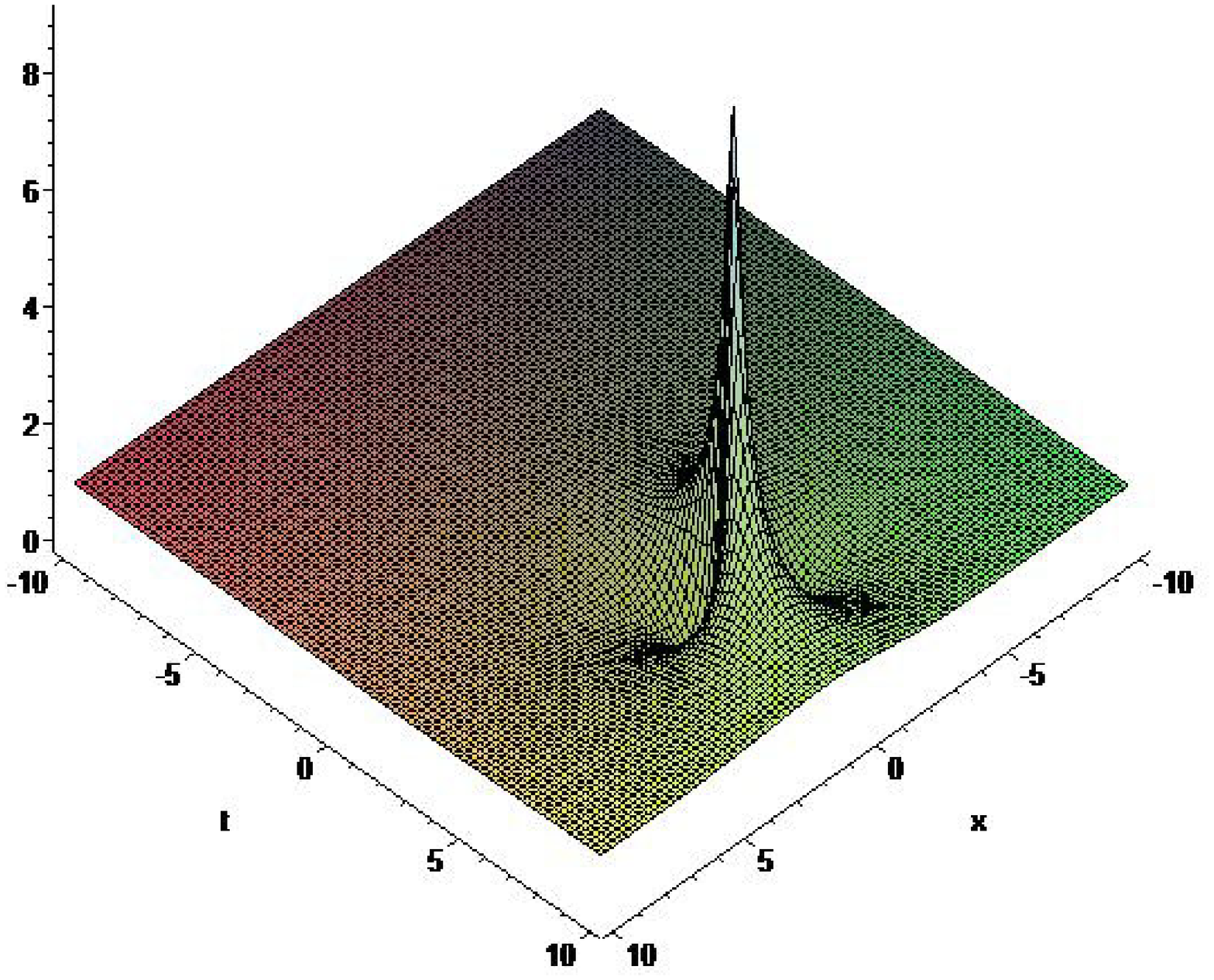}}
   \qquad
   \subfigure[]{\includegraphics[height=4cm,width=4cm]{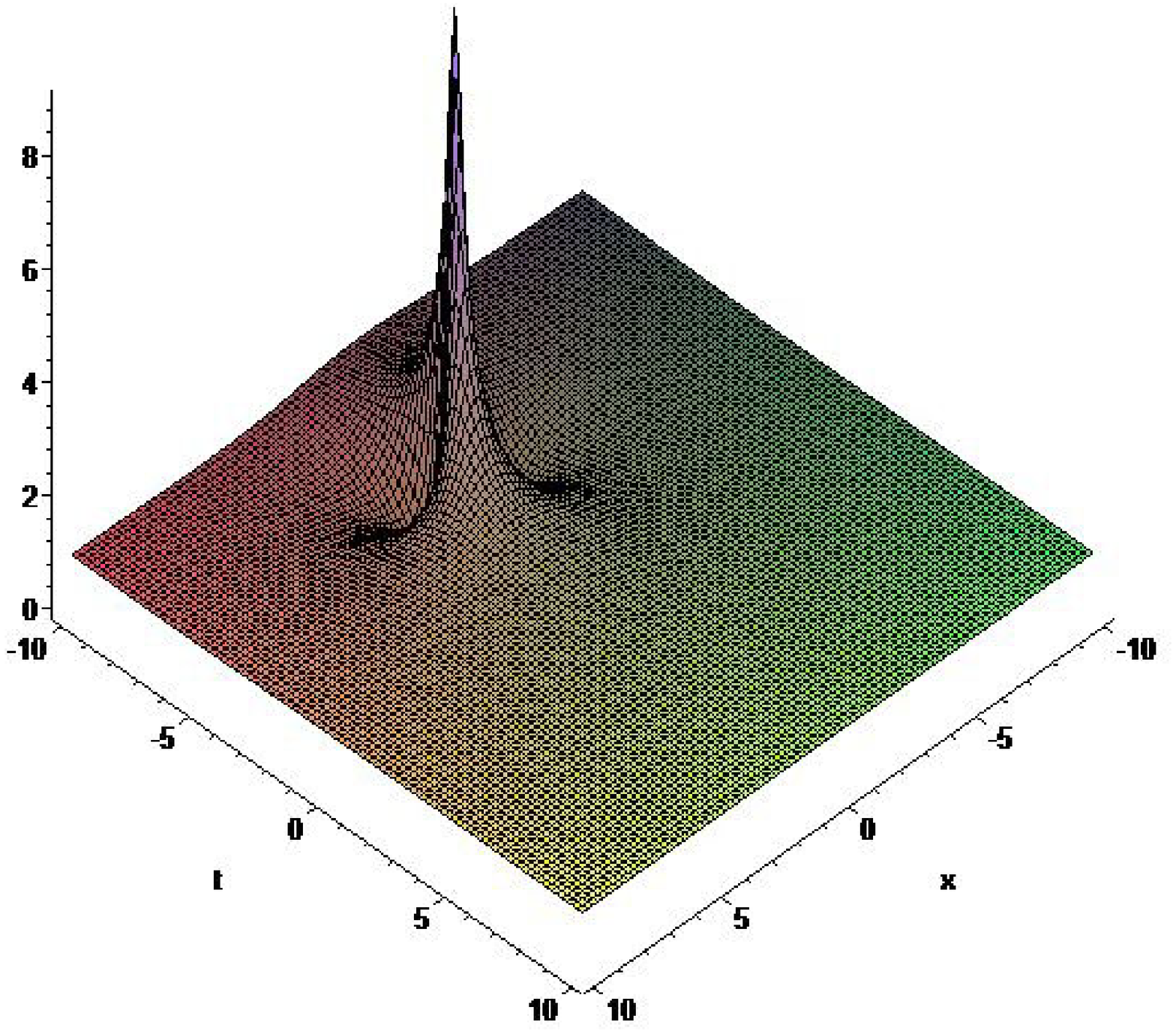}}
   \\
   \subfigure[]{\includegraphics[height=4cm,width=4cm]{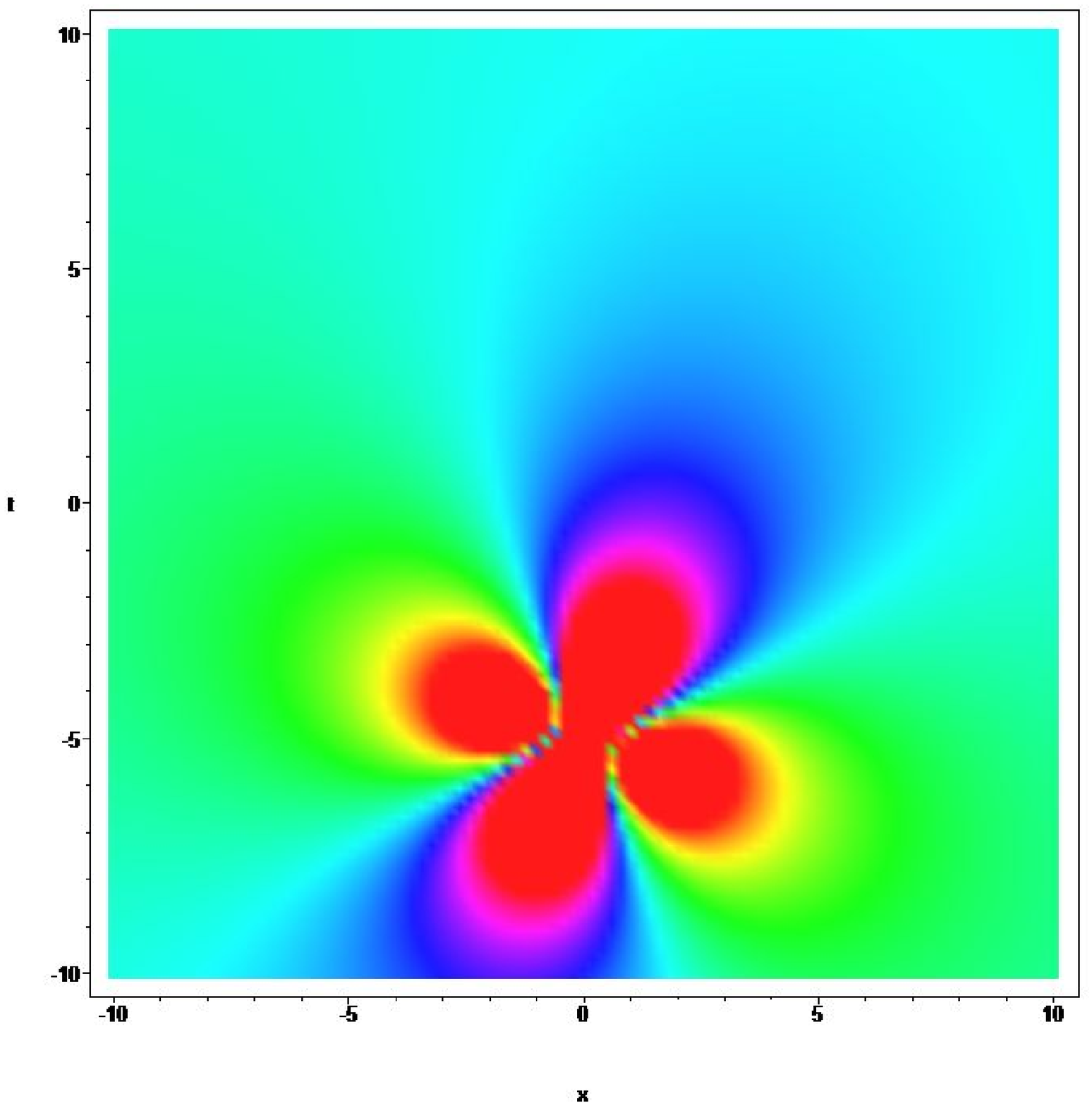}}
   \qquad
   \subfigure[]{\includegraphics[height=4cm,width=4cm]{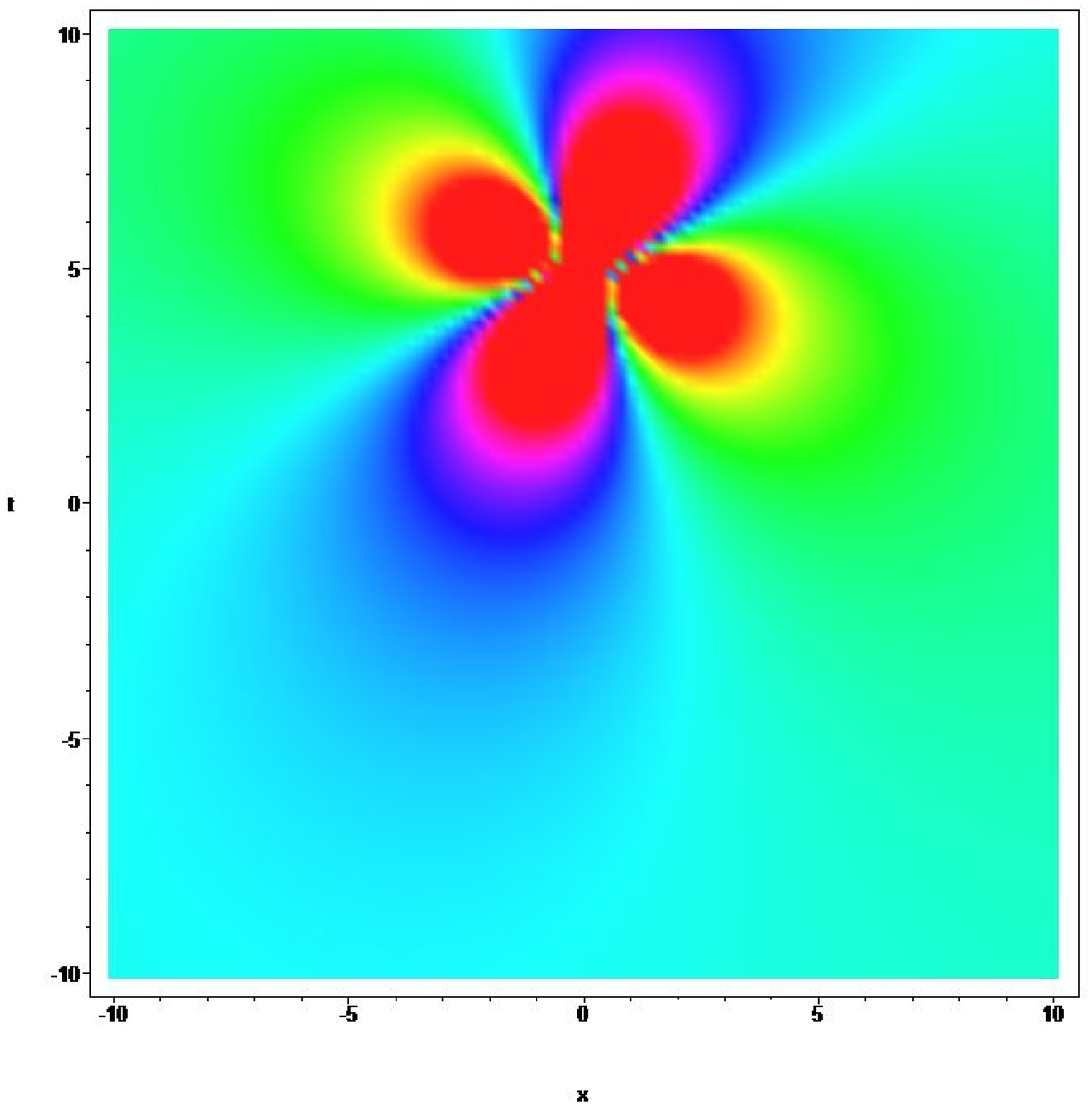}}
   \caption{(Color online)Various form of $1$-order rogue wave solution with particular parameter $S_0$. (a) The first order rogue wave with $S_0=5$, the maximum amplitude occurs at $x=0$ and $t=-5$. (b) The first order rogue wave with $S_0=-5$, the maximum amplitude occurs at $x=0$ and $t=5$}\label{fig.1rws0}
\end{figure}

\begin{figure}[!htp]
   \centering
   \subfigure[$2$-order rogue wave with $S_1=500$]{\includegraphics[height=4cm,width=4cm]{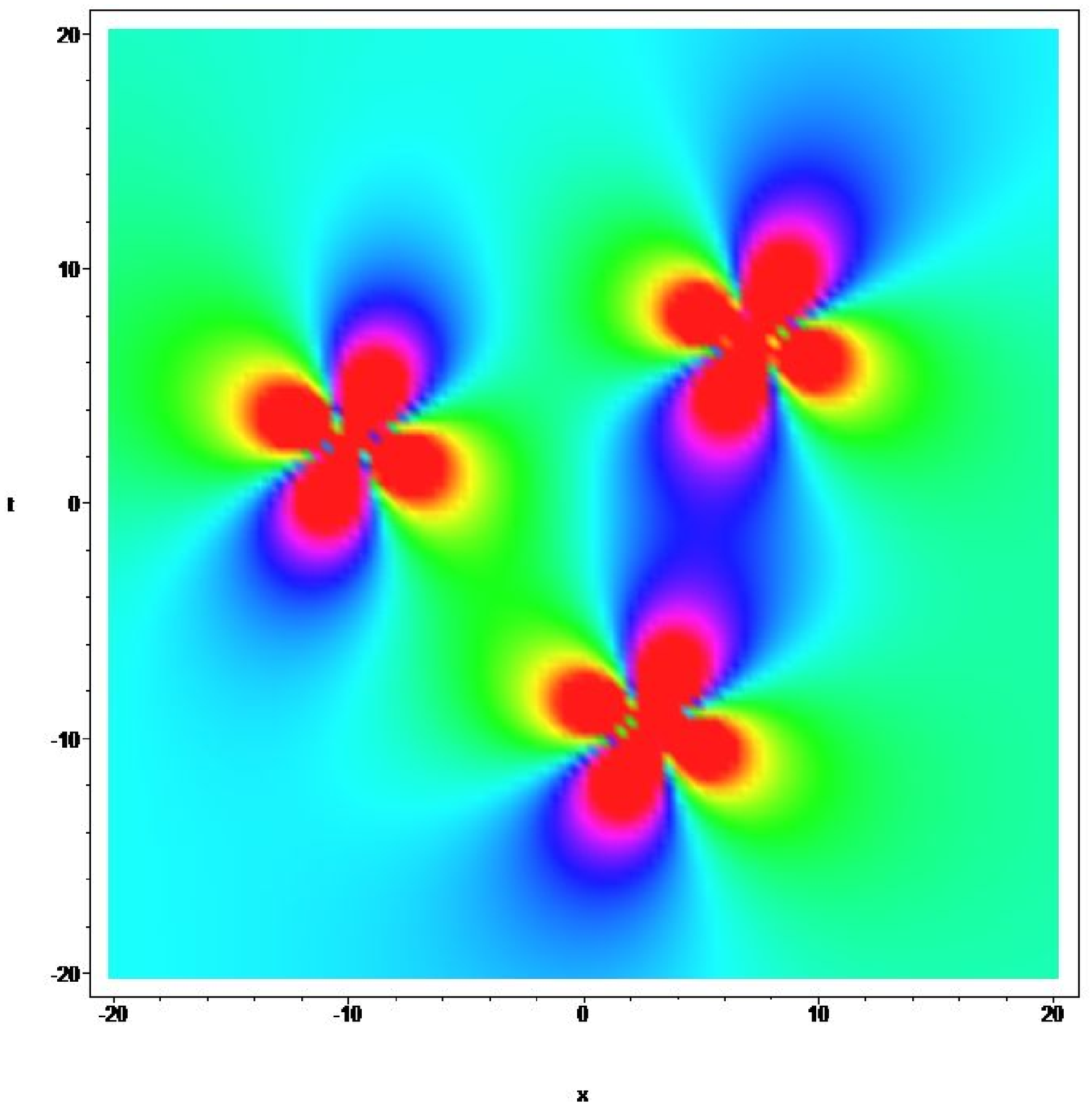}}
   \qquad
   \subfigure[$3$-order rogue wave with $S_1=500$]{\includegraphics[height=4cm,width=4cm]{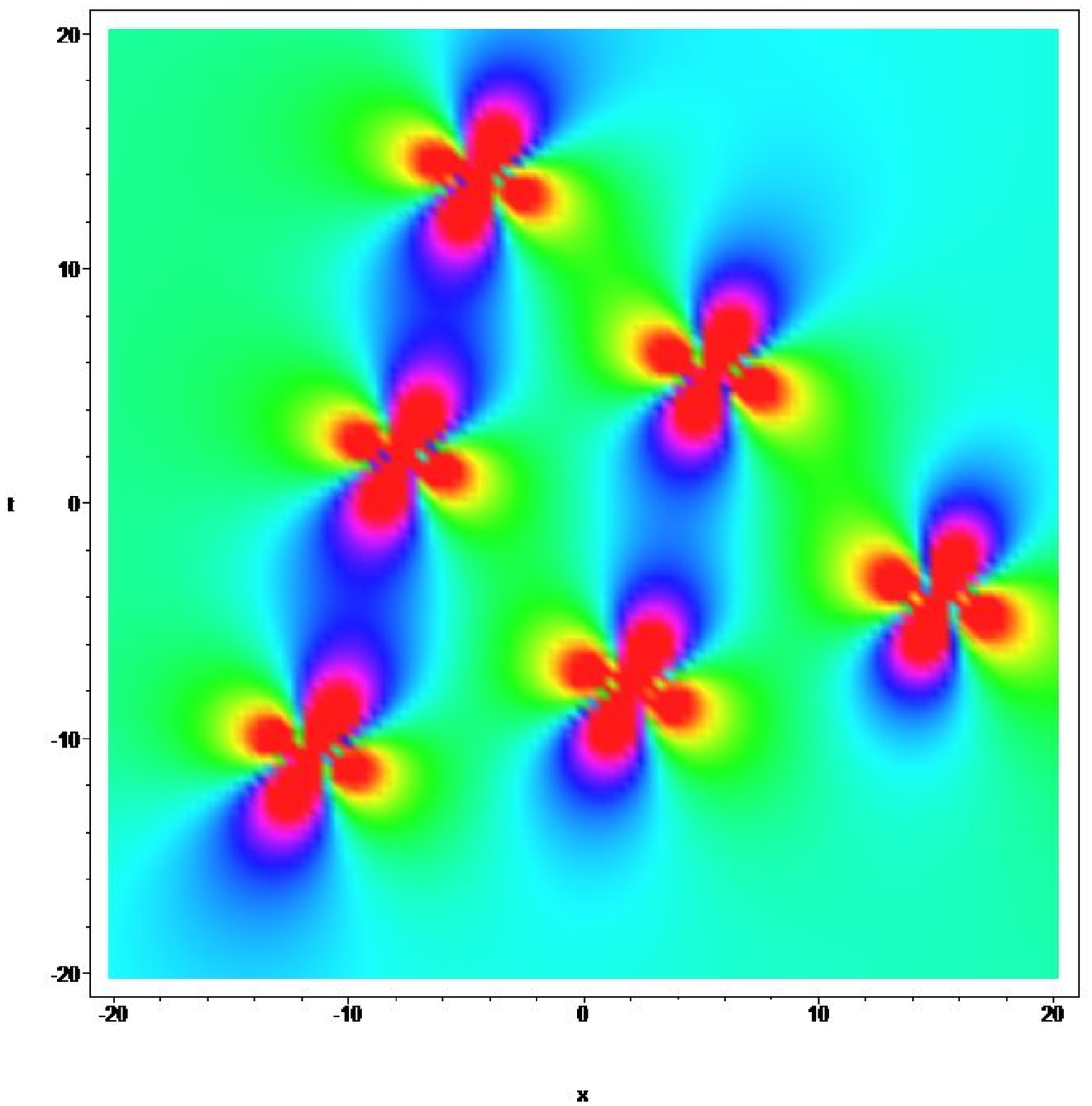}}
   \qquad
   \subfigure[$4$-order rogue wave with $S_1=500$]{\includegraphics[height=4cm,width=4cm]{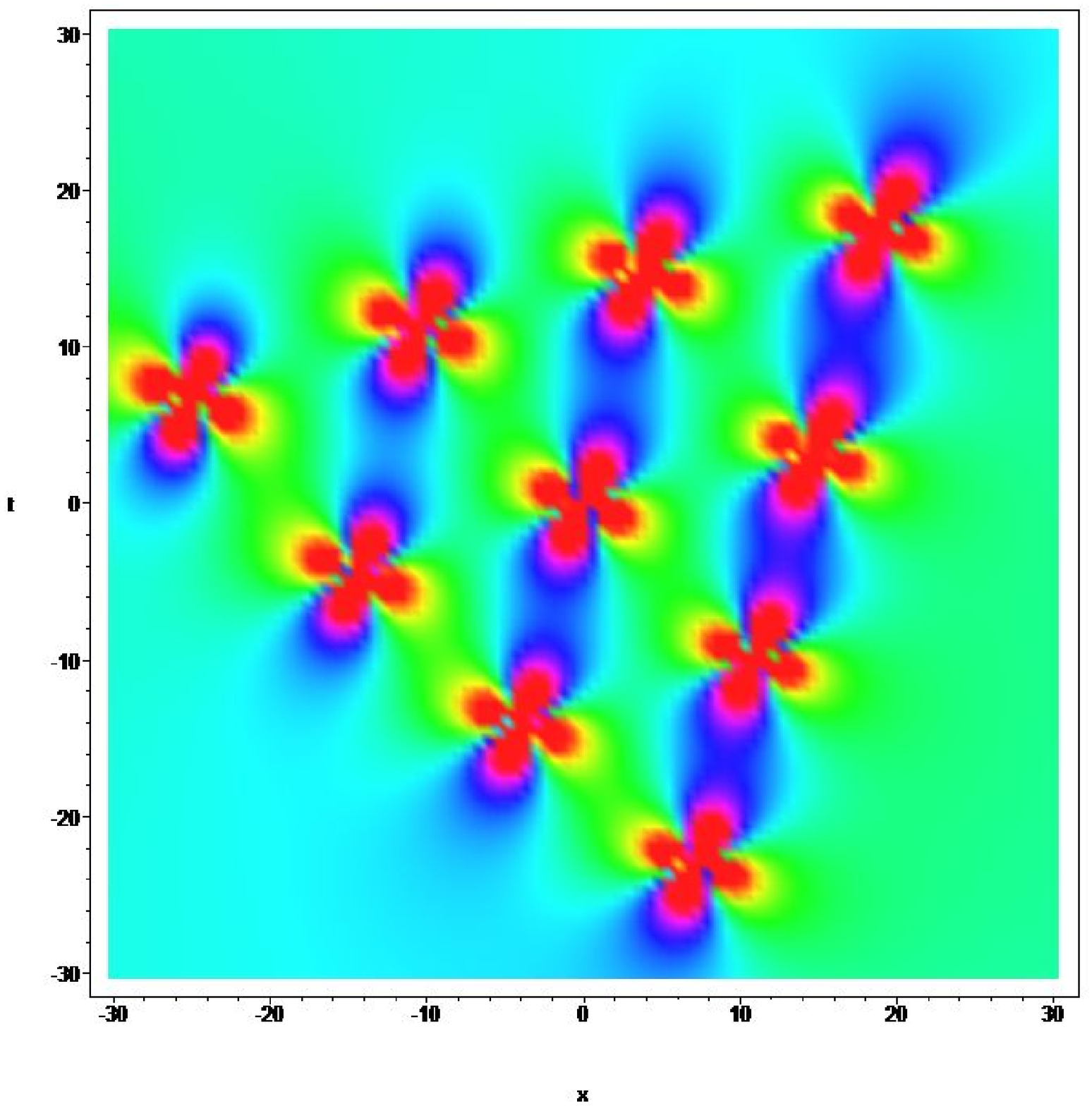}}
   \qquad
   \subfigure[$5$-order rogue wave with $S_1=500$]{\includegraphics[height=4cm,width=4cm]{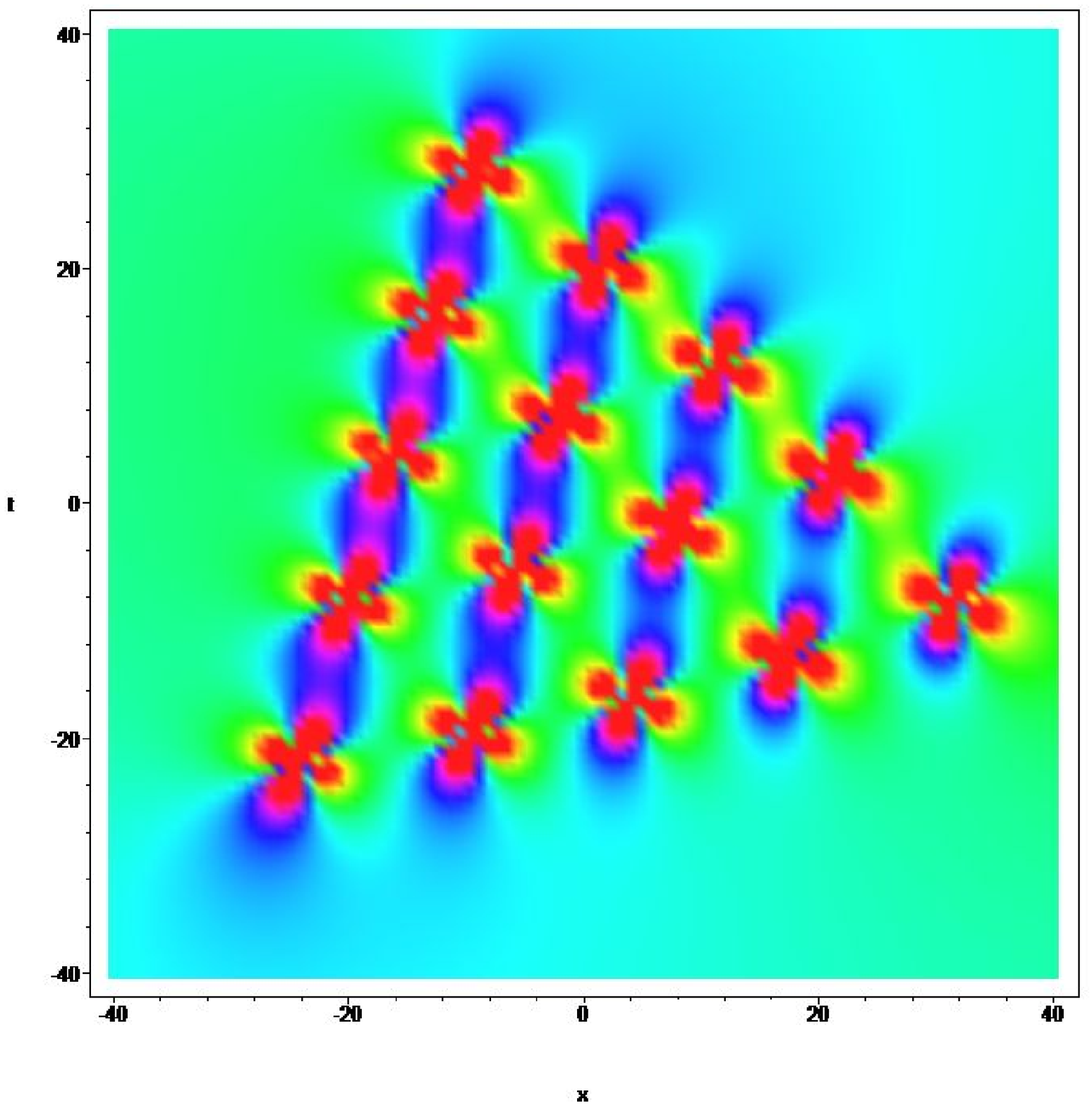}}
   \qquad
   \subfigure[$6$-order rogue wave with $S_1=500$]{\includegraphics[height=4cm,width=4cm]{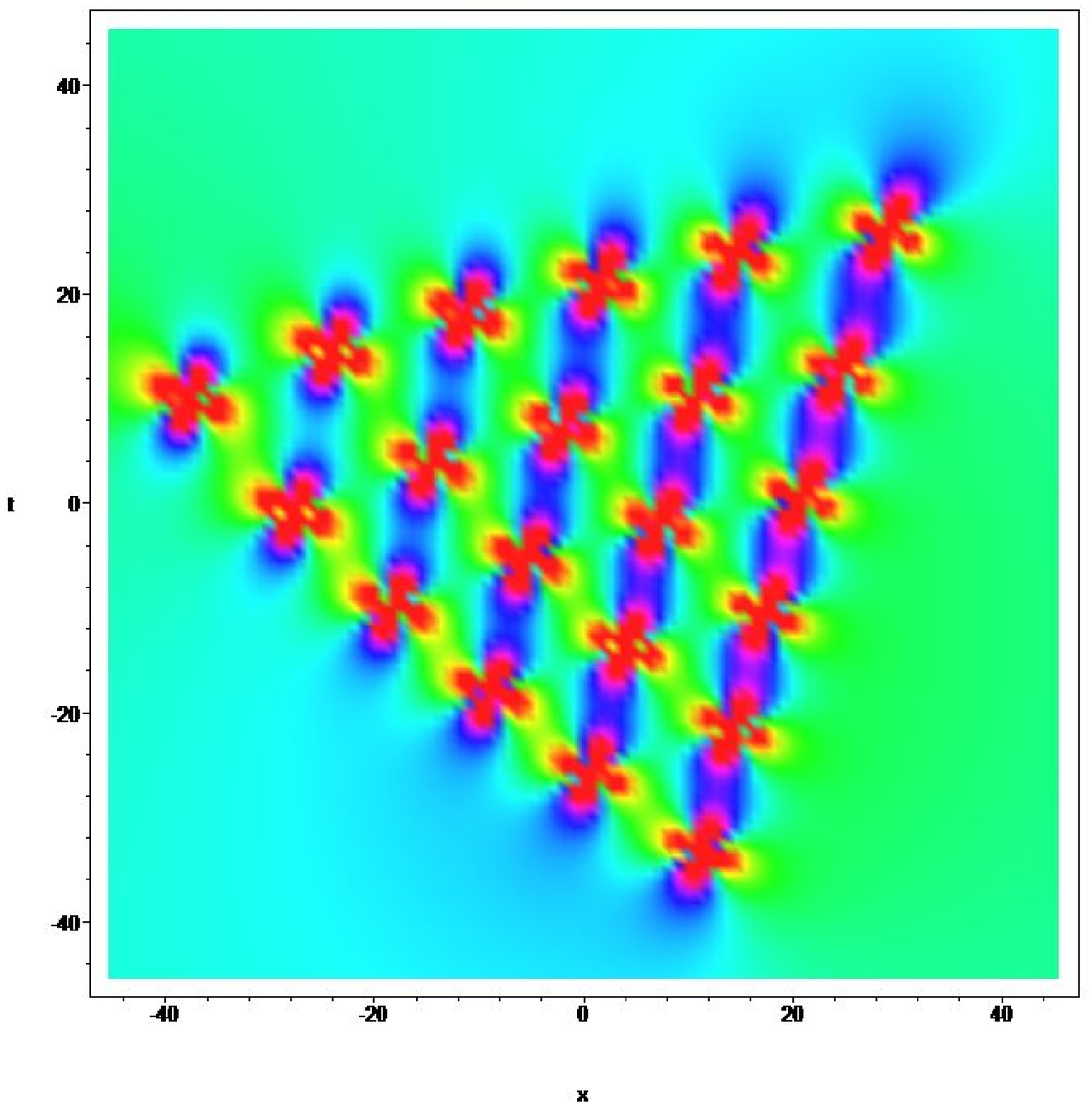}}
   \qquad
   \subfigure[$7$-order rogue wave with $S_1=250$]{\includegraphics[height=4cm,width=4cm]{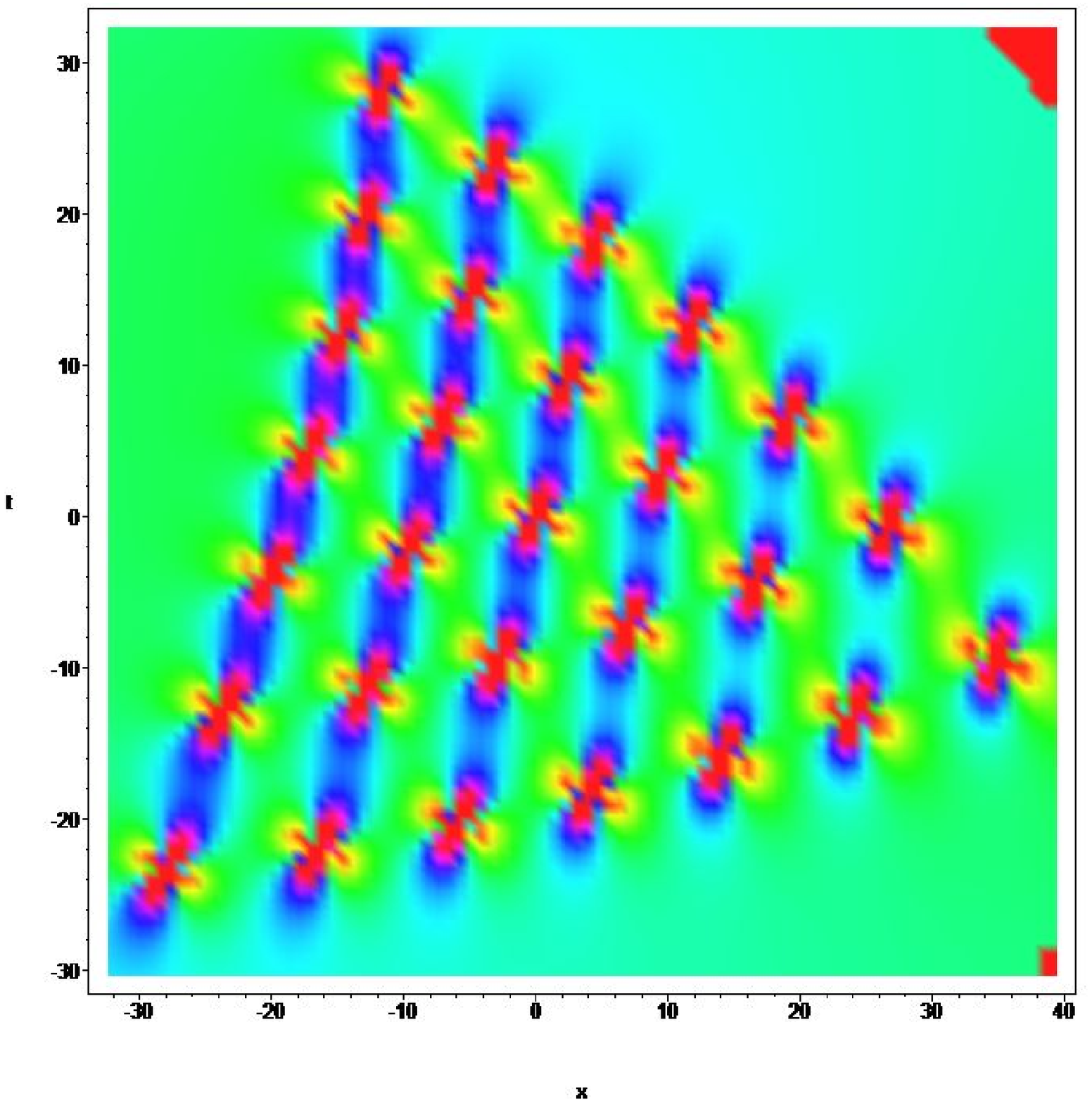}}
   \caption{The \emph{triangular} structures of higher order rogue waves.}\label{fig.triangle}
\end{figure}


\begin{figure}[!htp]
   \centering
   \subfigure[]{\includegraphics[height=4cm, width=4cm]{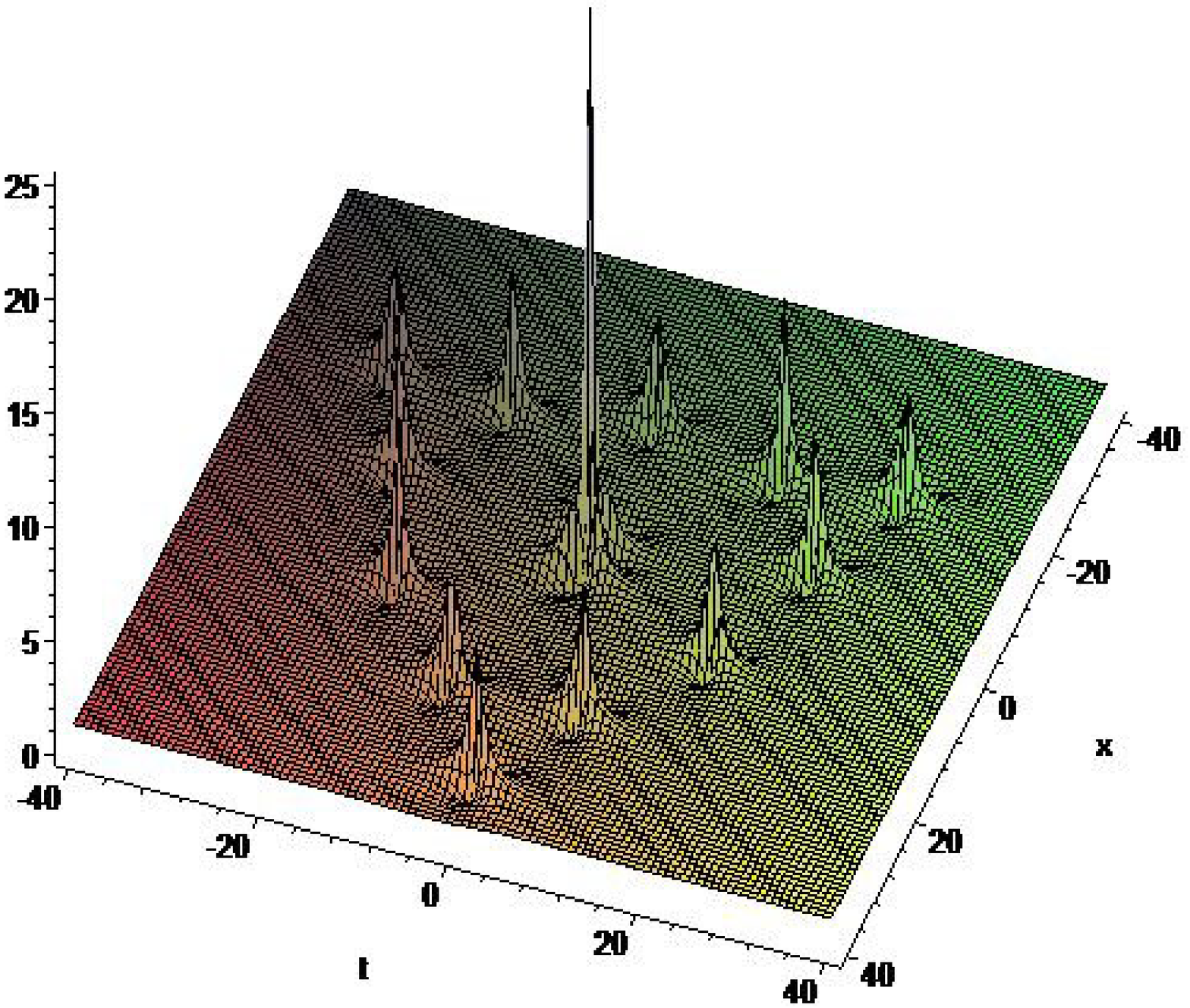}}
   \qquad
   \subfigure[]{\includegraphics[height=4cm, width=4cm]{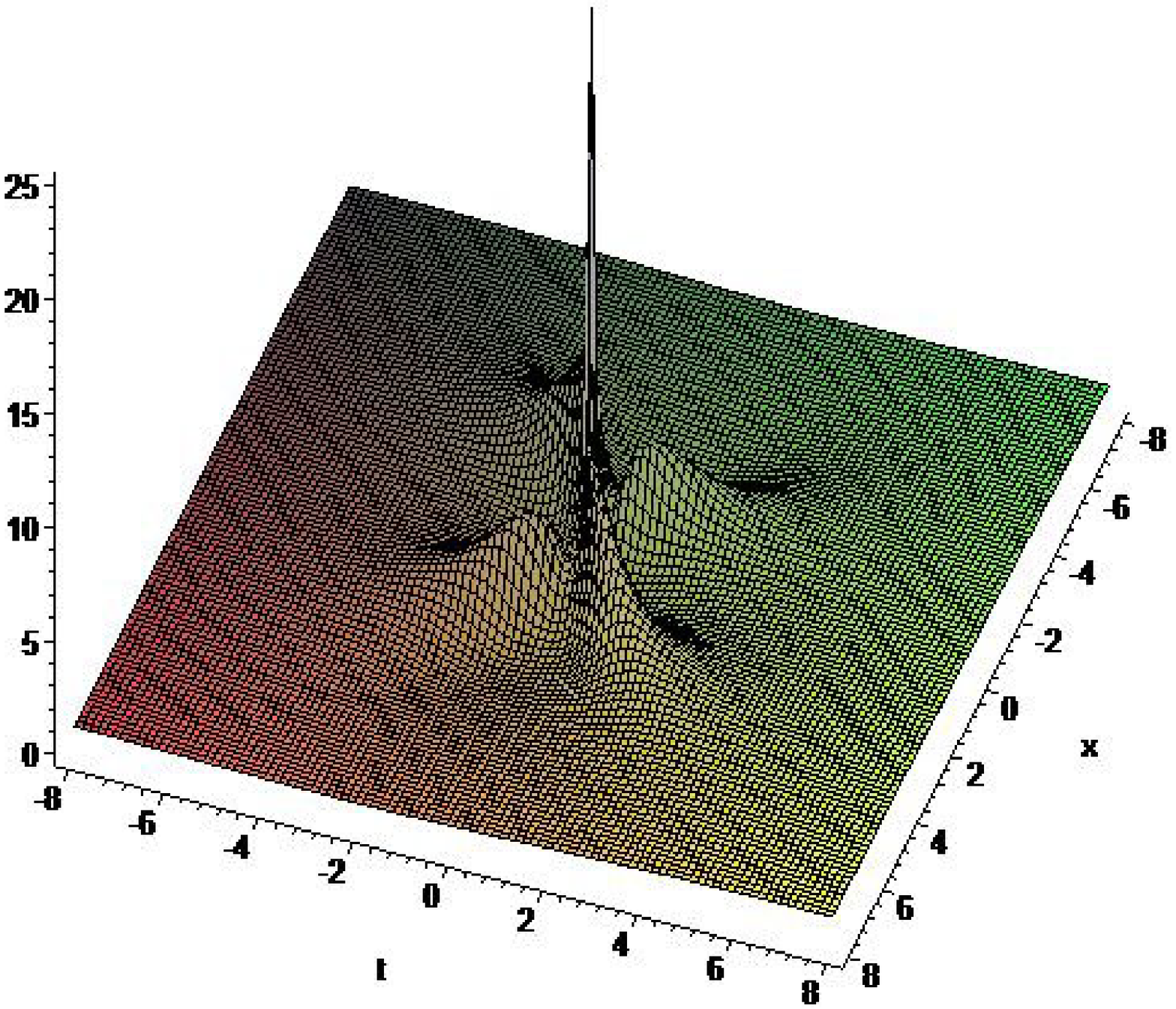}}
   \caption{(Color online)The \emph{modified-triangular} structure of $5$-order rogue wave with a second order rogue wave located in the center. (a) An overall profile of $5$-order rogue wave with $S_1=100$. (b) The centra profile of the right panel.}\label{fig.newtriangle}
\end{figure}

\begin{figure}[!htp]
   \centering
   \subfigure[$4$-order rogue wave with $S_3=10000$]{\includegraphics[height=4cm,width=4cm]{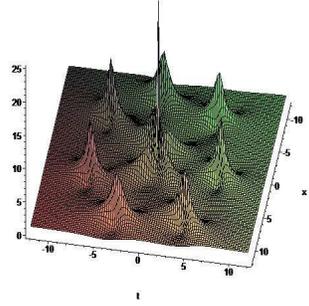}}
   \qquad
   \subfigure[$5$-order rogue wave with $S_4=500000$]{\includegraphics[height=4cm,width=4cm]{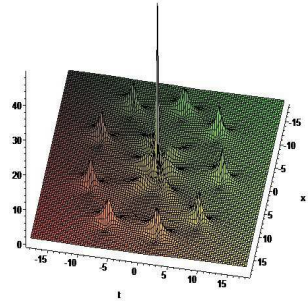}}
   \\
   \subfigure[$6$-order rogue wave with $S_5=1\times10^8$]{\includegraphics[height=4cm,width=4cm]{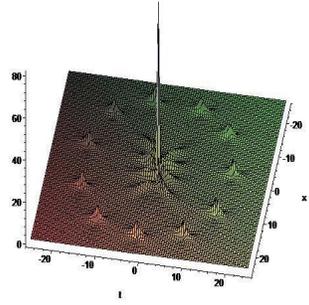}}
   \qquad
   \subfigure[$7$-order rogue wave with $S_6=1\times10^{10}$]{\includegraphics[height=4cm,width=4cm]{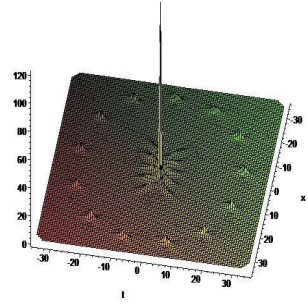}}
   \caption{(Color online) The \emph{ring} structures of higher order rogue wave solutions with $S_{k-1}\neq0$.}\label{fig.circle}
\end{figure}


\begin{figure}[!htp]
\centering
\subfigure[$4$-order rogue wave with $S_1=300$ and $S_3=1\times10^7$]{\includegraphics[height=4cm,width=4cm]{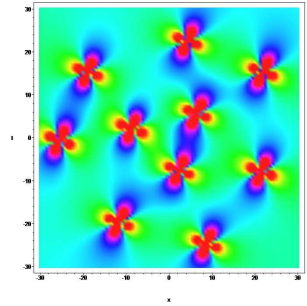}}
\qquad
\subfigure[$5$-order rogue wave with $S_1=200$ and $S_4=1\times10^9$]{\includegraphics[height=4cm,width=4cm]{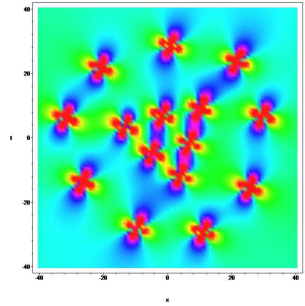}}
\\
\subfigure[$6$-order rogue wave with $S_1=500$ and $S_5=3\times10^{10}$]{\includegraphics[height=4cm,width=4cm]{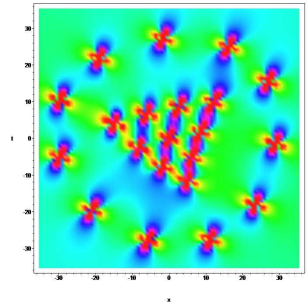}}
\qquad
\subfigure[$7$-order rogue wave with $S_1=300$ and $S_6=5\times10^{10}$]{\includegraphics[height=4cm,width=4cm]{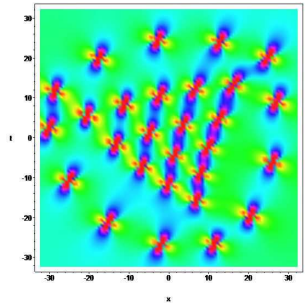}}
\caption{(Color online) A \emph{triangular} pattern in a ring for higher order rogue wave with $S_{k-1}\neq0$ and $S_1\neq0$.}\label{fig.circletriangle}
\end{figure}

\begin{figure}[!htp]
\centering
\subfigure[]{\includegraphics[height=4cm,width=4cm]{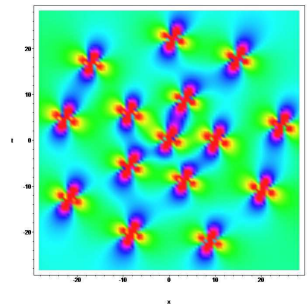}}
\qquad
\subfigure[]{\includegraphics[height=4cm,width=4cm]{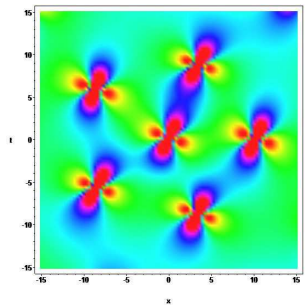}}
\caption{(Color online) (a) The \emph{multi-ring} model of  $5$-order rogue wave solution with $S_2=5000$ and $S_4=1\times10^8$. (b)
The local central profile.}\label{fig.5rwmulticircle}
\end{figure}


\begin{figure}[!ht]
\centering
\subfigure[The entire part]{\includegraphics[height=4cm,width=4cm]{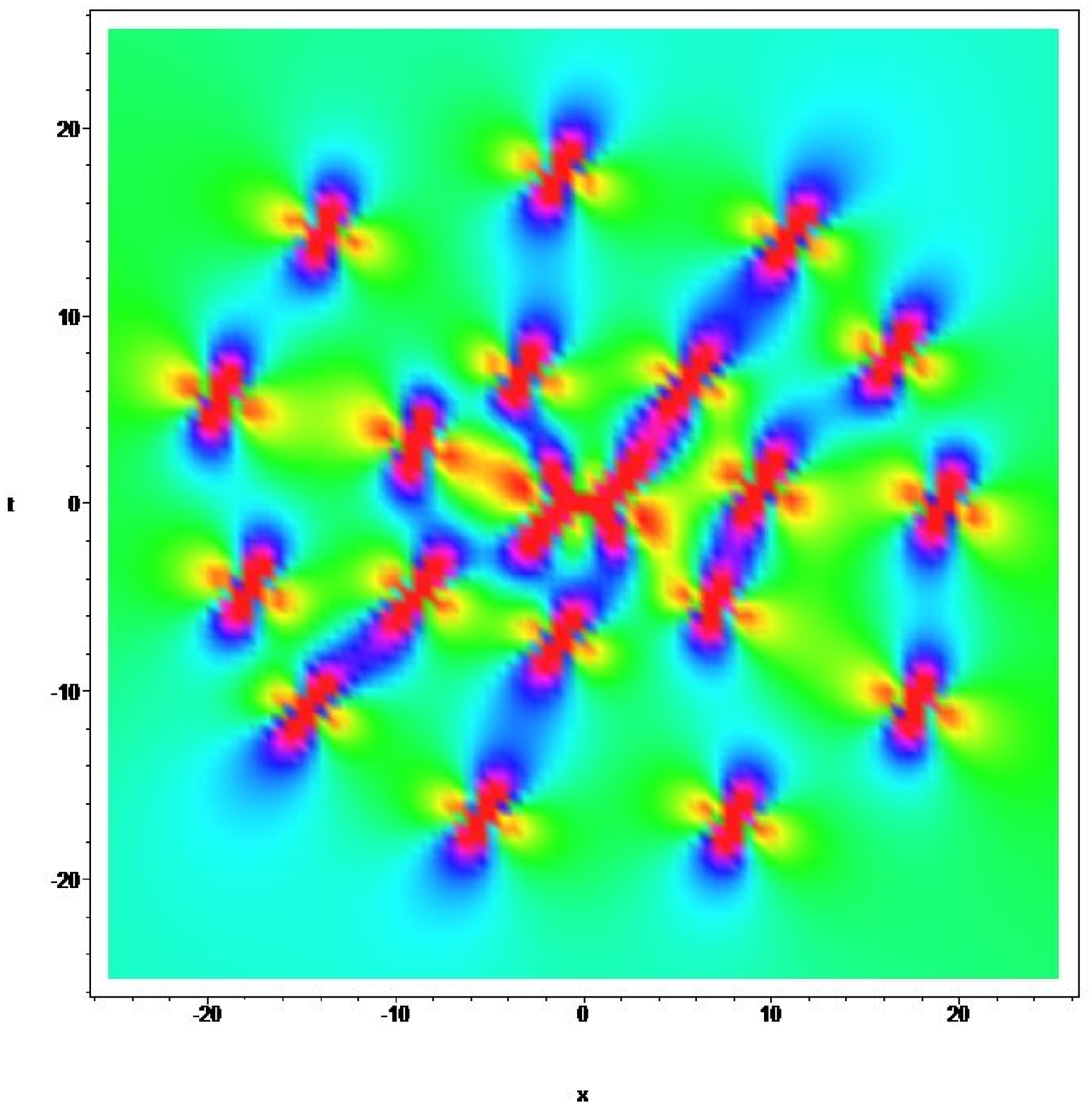}}
\qquad
\subfigure[The inner part]{\includegraphics[height=4cm,width=4cm]{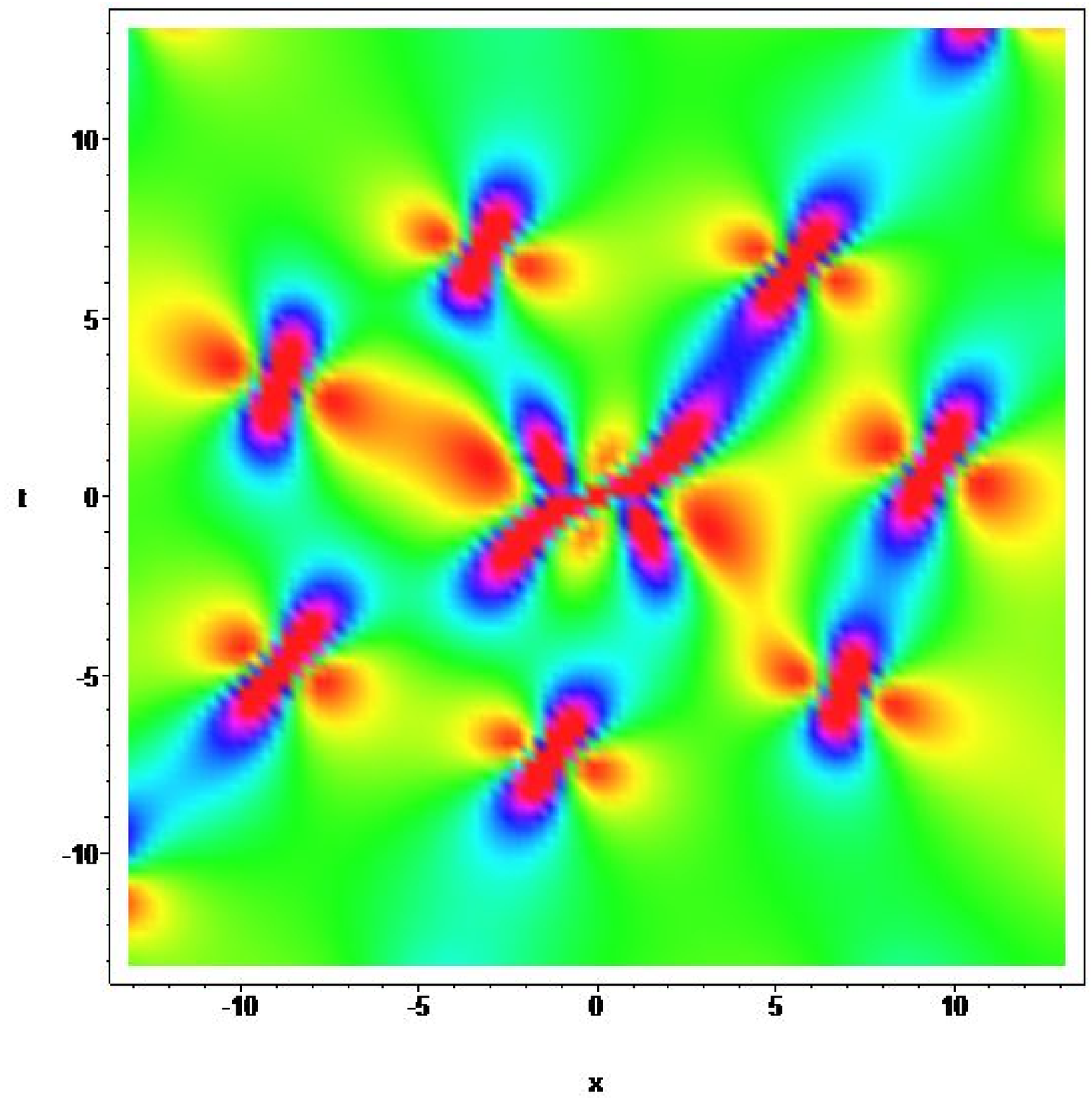}}
\caption{(Color online)The \emph{multi-ring-1} structure of $6$-order rogue wave .}\label{fig.6rwmulticircle1}
\end{figure}

\begin{figure}[!ht]
\centering
\subfigure[The entire part]{\includegraphics[height=4cm,width=4cm]{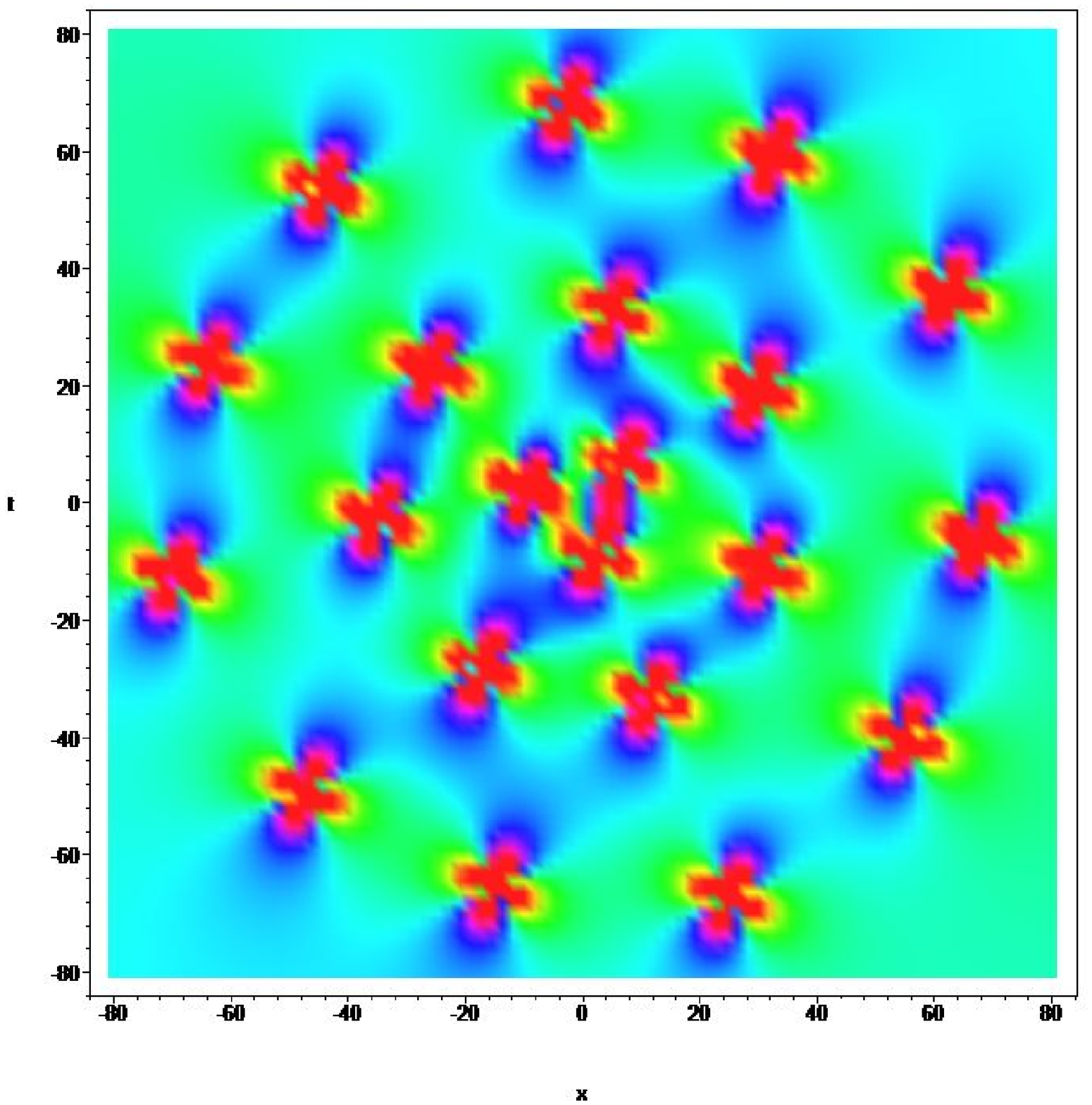}}
\qquad
\subfigure[The inner part]{\includegraphics[height=4cm,width=4cm]{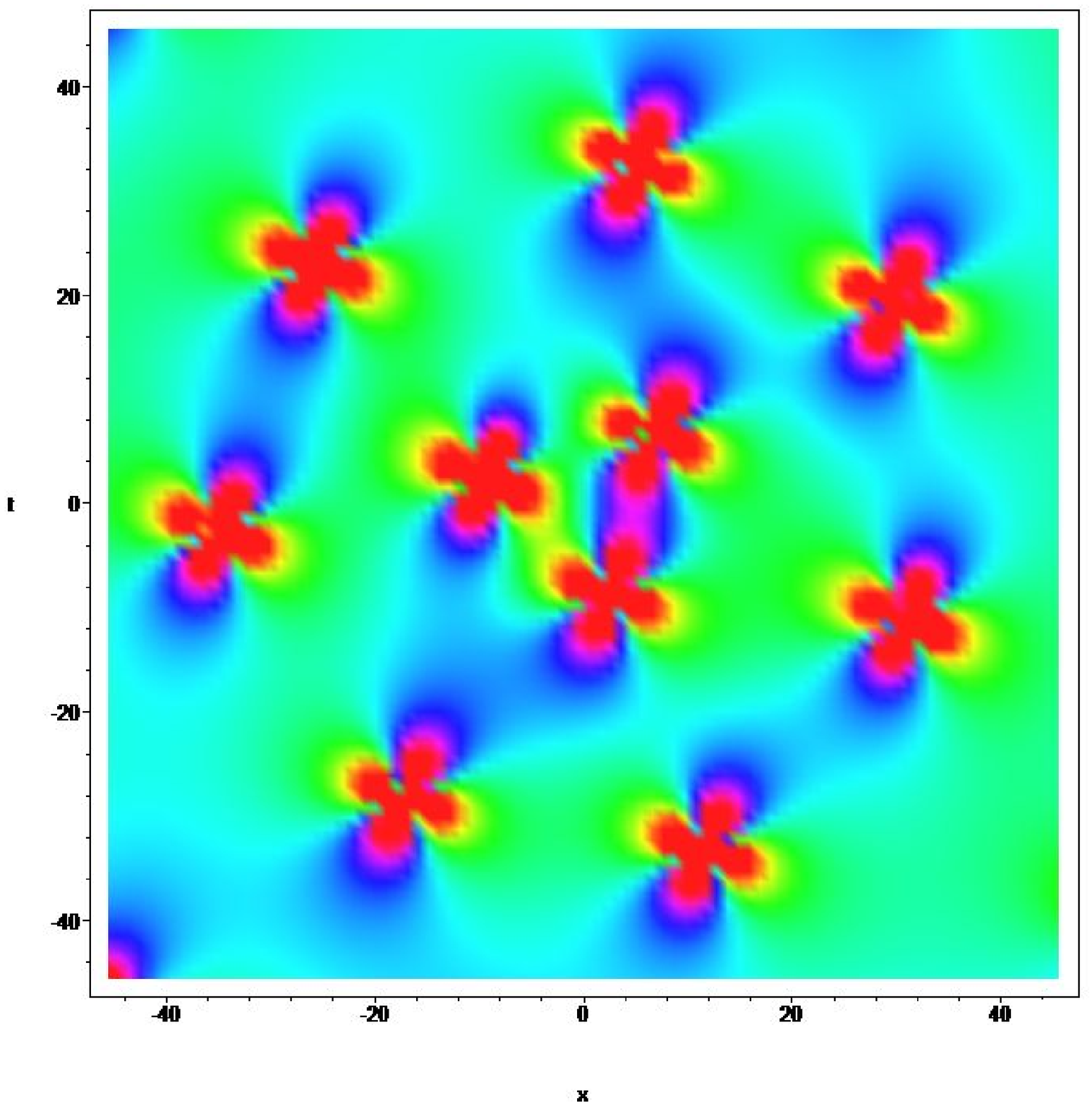}}
\caption{(Color online)The \emph{multi-ring-2} structure of the $6$-order rogue wave. Both the outer shell and the middle shell are circular, and the inner shell is triangular(or circular).}\label{fig.6rwmulticircle2}
\end{figure}


\begin{figure}[!htp]
\centering
\subfigure[The entire part]{\includegraphics[height=4cm,width=4cm]{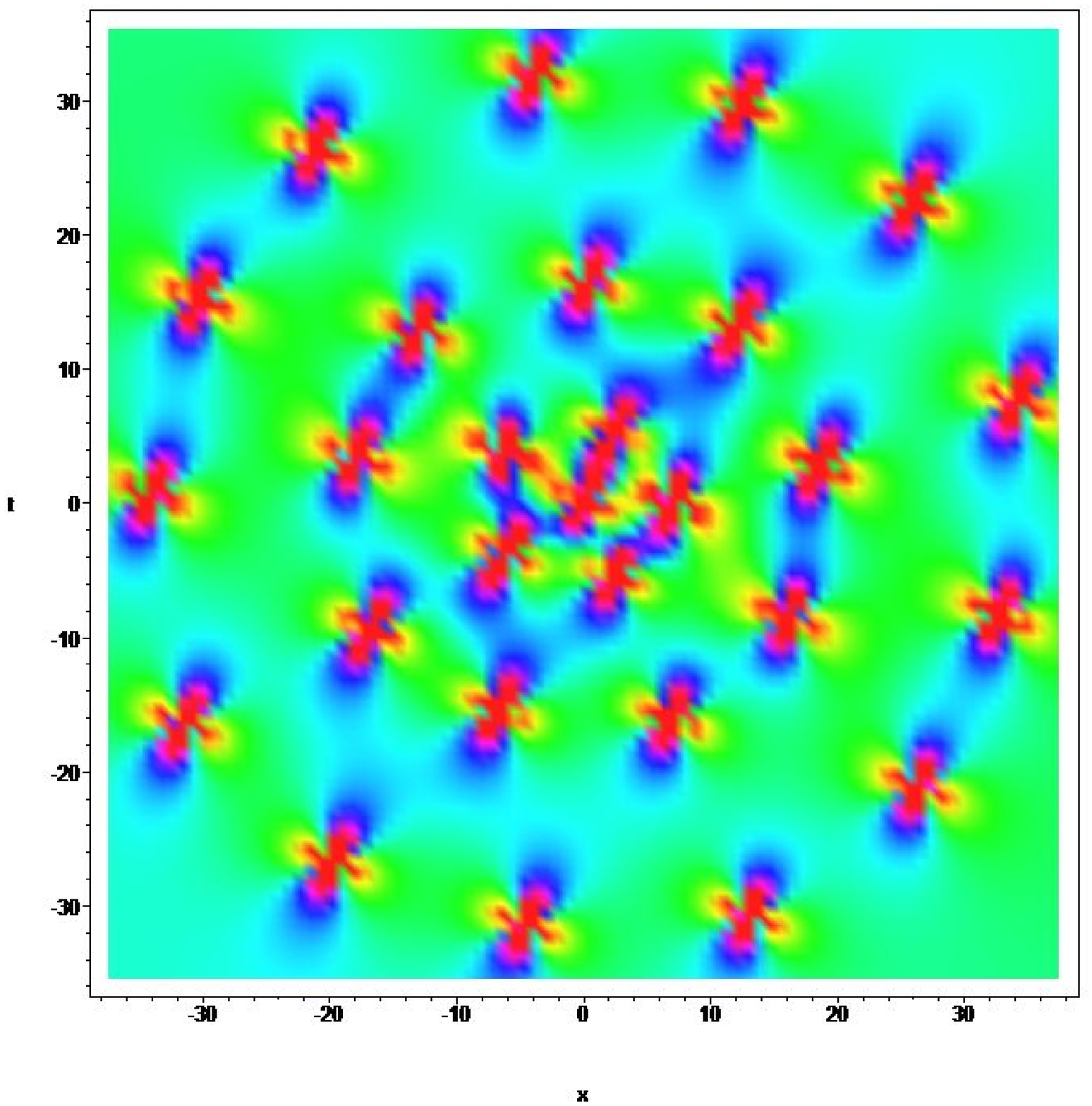}}
\qquad
\subfigure[The inner part]{\includegraphics[height=4cm,width=4cm]{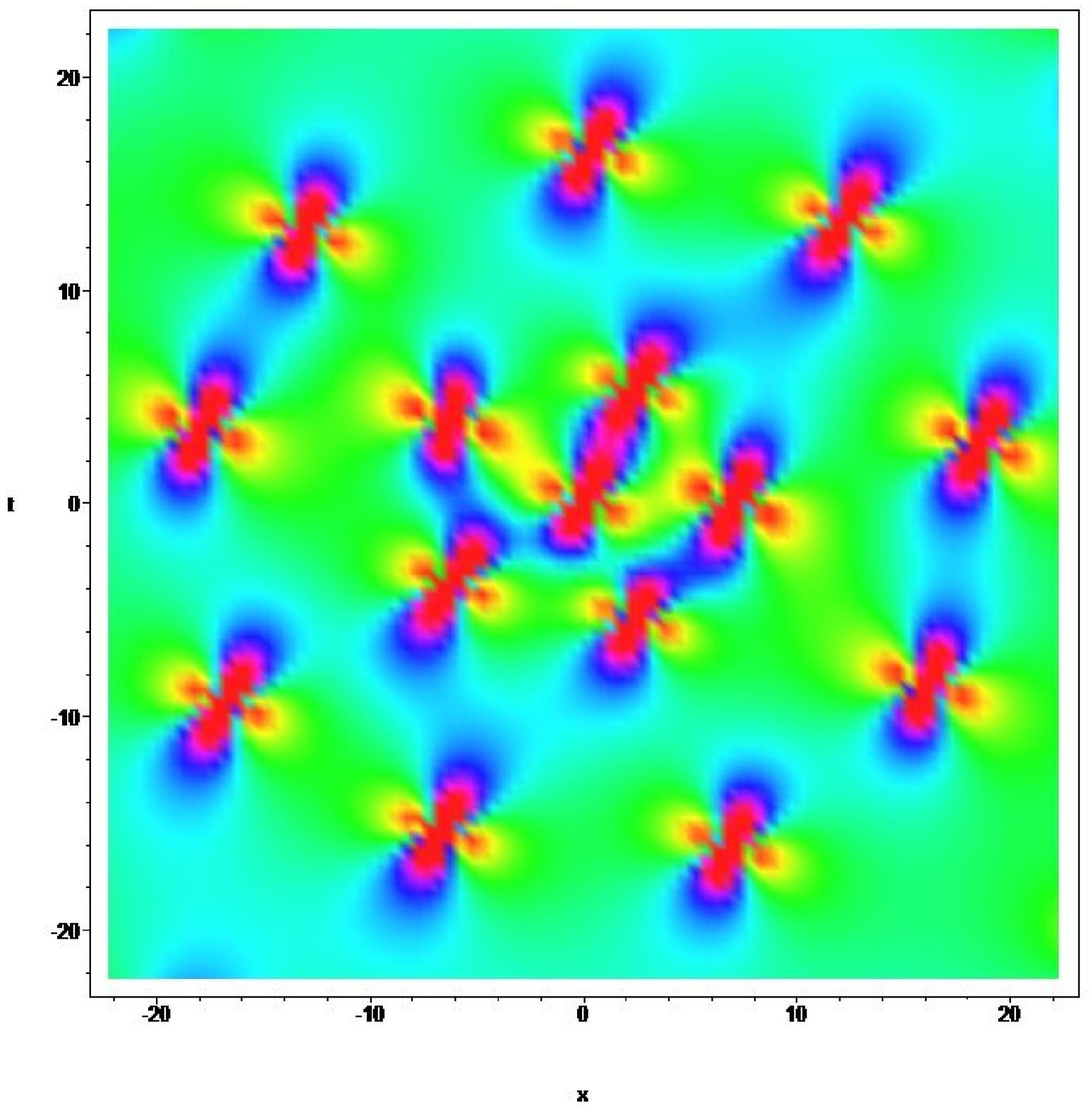}}
\caption{ (Color online)  The \emph{multi-ring-1} structure of $7$-order rogue wave with $S_2 = 1000$, $S_4 =
1\times10^7$and $S_6= 5\times10^{11}$. All the three shell are circular.}\label{7rwmulticircle1}
\end{figure}

\begin{figure}[!htp]
\centering
\subfigure[The entire part]{\includegraphics[height=4cm,width=4cm]{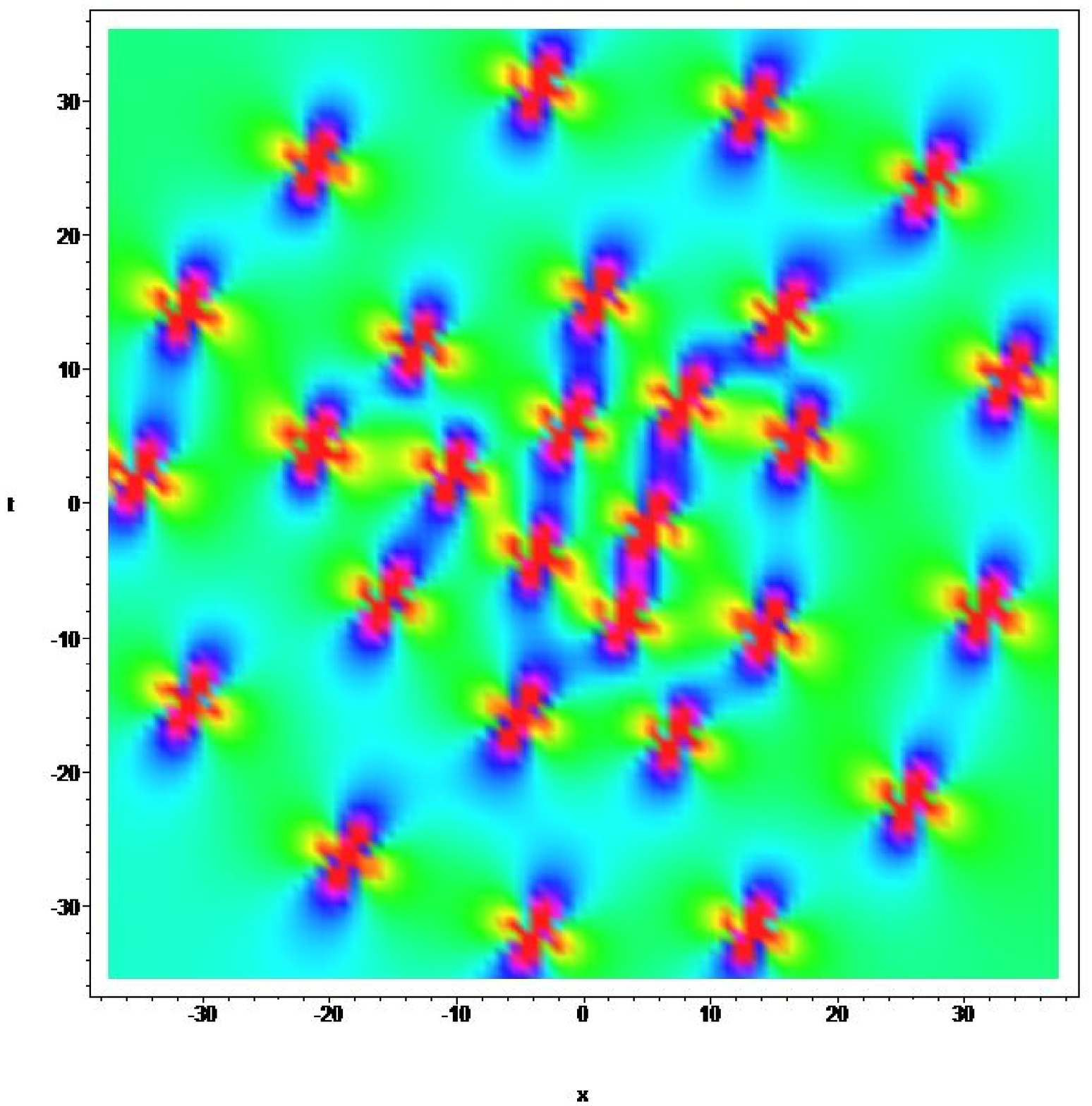}}
\qquad
\subfigure[The inner part]{\includegraphics[height=4cm,width=4cm]{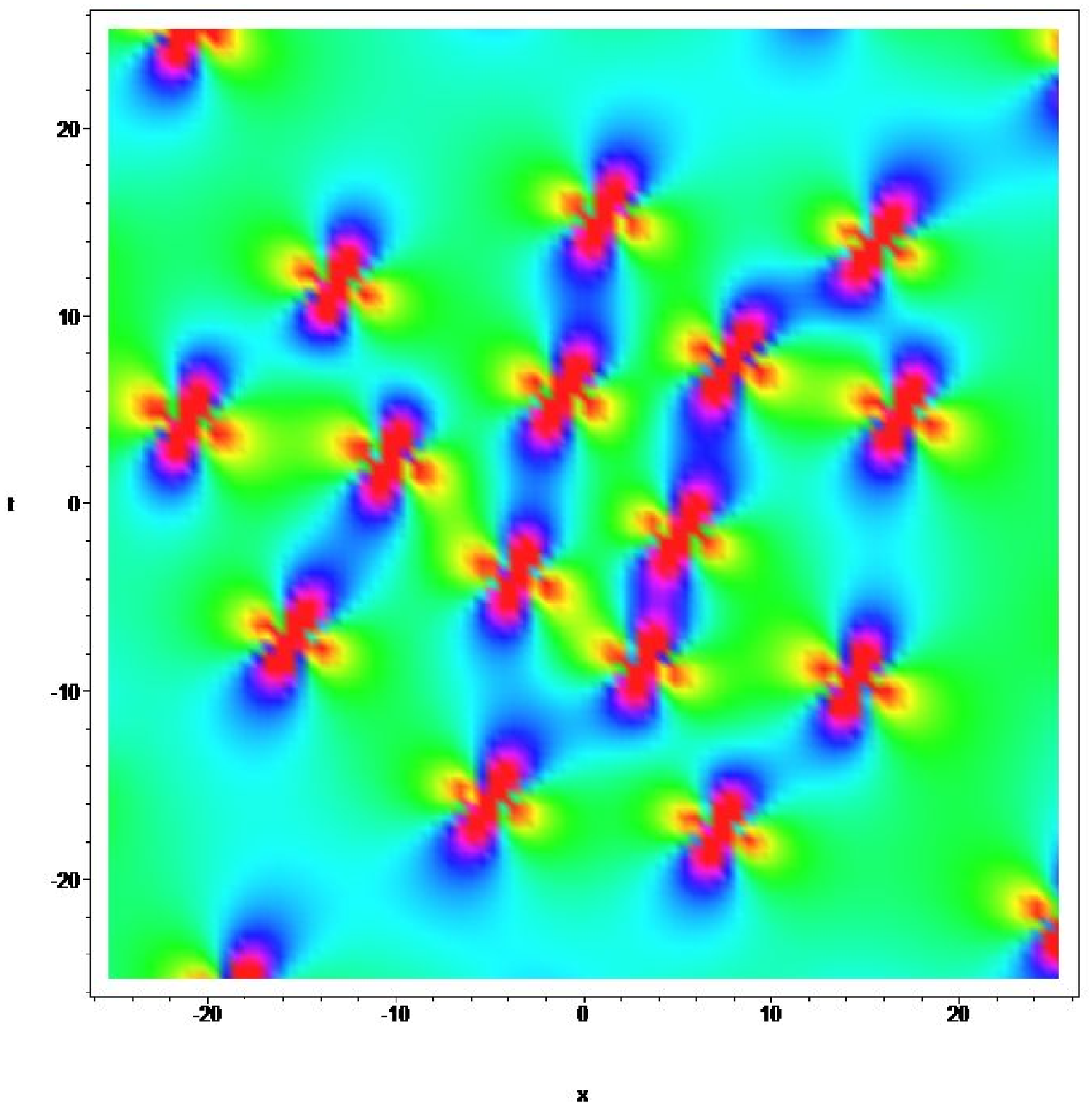}}
\caption{ (Color online) The \emph{multi-ring-2} structure of $7$-order rogue wave with $S_1 = 100$, $S_4 =
1\times10^7$and $S_6= 5\times10^{11}$. Both the outer shell and the middle shell are circular, and the inner shell is triangular.}\label{7rwmulticircle2}
\end{figure}


\end{document}